%% file: arxiv1.tex
\documentclass{article}

\usepackage{amsmath,amssymb,amsthm,fullpage,mathrsfs,pgf,tikz,caption,subcaption,mathtools,mathabx}
\usepackage[ruled,vlined,noresetcount]{algorithm2e}
\usepackage{amsmath,amssymb,amsthm,mathtools}
\usepackage{thmtools}
\usepackage[utf8]{inputenc} 
\usepackage[T1]{fontenc}    
\usepackage{hyperref}       
\usepackage{url}            
\usepackage{booktabs}       
\usepackage{amsfonts}       
\usepackage{nicefrac}       
\usepackage{microtype}      
\usepackage{times}
\usepackage{bbm}
\usepackage{enumitem}
\usepackage{xcolor}
\usepackage{cleveref}
\usepackage{bm}
\usepackage{float}
\usepackage[margin=1in]{geometry}
\usepackage{tikz-cd}

\input{macros.tex}

\title{Transversal non-Clifford gates for quantum LDPC codes on sheaves}
\author{Ting-Chun Lin\thanks{Department of Physics, University of California San Diego, CA, and Hon Hai Research Institute, Taipei, Taiwan. Email: \texttt{til022@ucsd.edu}.}}

\begin{document}

\sloppy

\maketitle

\begin{abstract}
  A major goal in quantum computing
    is to build a fault-tolerant quantum computer.
  One approach involves quantum low-density parity-check (qLDPC) codes
    that support transversal non-Clifford gates.
  In this work,
    we provide a large family of such codes.
  The key insight is to interpret the logical operators of qLDPC codes
    as geometric surfaces
    and use the intersection number of these surfaces to define the non-Clifford operation.
  At a more abstract level,
    this construction is based on defining the cup product on the chain complex induced from a sheaf.
\end{abstract}

\section{Introduction}

A major goal in quantum computing
  is to build a fault-tolerant quantum computer.
Without fault tolerance,
  running quantum algorithms (quantum circuits) directly
  leads to error propagation,
  and the resulting output becomes unreliable.
In contrast, a fault-tolerant version
  can achieve arbitrarily low error,
  provided that every gate has error below a certain threshold.
It has been proven that any quantum algorithm
  can be implemented in a fault-tolerant manner.
This result is known as the threshold theorem
  and was first demonstrated using concatenation of quantum codes
  \cite{aharonov1997fault,kitaev1997quantum,knill1998resilient}.

Although fault-tolerant circuits have reliable output,
  they often require more qubits and longer runtime.
These costs, known as the space and time overheads,
  are not only of theoretical importance
  but also poses practical challenges.
As a result, researchers have been exploring various fault-tolerant schemes
  each with different trade-offs and advantages.
Some notable techniques include
  code concatenation \cite{aharonov1997fault,kitaev1997quantum,knill1998resilient},
  magic state distillation \cite{bravyi2005universal},
  gate teleportation \cite{gottesman1999quantum},
  gauge fixing or code switching \cite{bombin2015gauge,anderson2014fault},
  transversal gates \cite{shor1996fault},
  and quantum low-density parity-check (qLDPC) codes \cite{gottesman2013fault}.
See \cite{gottesman2022opportunities} for a review.

One of those fault-tolerant schemes relies on qLDPC codes with transversal gates.
In this approach, the quantum state in the original circuit $|\psi\>$
  is encoded into a new state $\Enc(|\psi\>)$ using a qLDPC code.
Additionally, gate operations from the original circuit $U$
  are implemented as transversal operations $\Enc(U) = \bigotimes_i U_i$
    where $U_i$ acts on the $i$-th qubit,
  such that $\Enc(U |\psi\>) = \Enc(U) \Enc(|\psi\>)$.
The advantage of transversal operations is that they do not
  spread errors.
The advantage of qLDPC codes is that
  the syndromes can be extracted rather reliably in constant depth.
Therefore, every layer of gates in the original circuit
  is transformed into two layers
  where the first layer applies $\Enc(U)$ on $\Enc(|\psi\>)$,
  and the second layer runs the decoder to correct the error.
In many cases, this process keeps the errors under control
  which leads to a fault tolerant quantum circuit.
(In actuality, this scheme often requires code switching
  to achieve a universal gate set, but we will not discuss that aspect here.)

The challenge in implementing this idea
  lies in constructing qLDPC codes that support transversal gates.
While it is known that any quantum CSS code supports transversal CNOT gates,
  additional gates are required to obtain a universal gate set.
Some examples include:
  2D color codes which support transversal H gates,
  3D color codes which support transversal T gates \cite{bombin2007topological}
  and 3D toric codes which support transversal CCZ gates \cite{vasmer2019three}.
More recently, there have been efforts to construct codes from manifolds \cite{zhu2023non}
  and codes formed by gluing color codes \cite{vuillot2022quantum,scruby2024quantum}.
However, for these constructions, a key code parameter $k_{\CCZ}$,
  which quantifies the number of logical CCZ gates that can be applied,
  is currently unknown.
This highlights the need for further construction of qLDPC codes that support non-Clifford gates.

In this work, we provide a large family of qLDPC codes that supports transversal CCZ gates.
These are known as sheaf codes \cite{first2022good,panteleev2024maximally},
  and they encompass all recent constructions of qLDPC codes \cite{panteleev2021asymptotically,leverrier2022quantum,dinur2022good,dinur2023new,dinur2024expansion}.
In a companion paper \cite{golowich2024quantum},
  we present an instance of this family that achieves
  $[[n, k \ge n^{1-\epsilon}, k_{\CCZ} \ge n^{1-\epsilon}, d \ge n^{1/3} / \poly\log n]]$ over qubits
  with $\poly\log n$ locality
  for arbitrarily small $\epsilon > 0$.
In the context of magic state distillation,
  this means that,
  given $n$ noisy CCZ resource states,
  the code can output $k_{\CCZ}$ more reliable CCZ states.
This leads to a magic state distillation protocol
  with $\gamma = \log(n/k_{\CCZ}) / \log(d) \to 0$.
In comparison,
  recent magic state distillation protocols with $\gamma \to 0$
  \cite{wills2024constant,golowich2024asymptotically,nguyen2024good}
  rely on codes with stabilizers that have linear weight.

\begin{remark}
  We note an upcoming paper with overlapping results \cite{breuckmann2024cups}.
  We will discuss the similarities and differences in the next version.
\end{remark}

\subsection{Main result}

To explain the result, we first briefly review the notion of sheaves over cellular complexes.
Given a cellular complex $X$ of dimension $t$,
  the traditional cochain complex $\cC^{\bullet}(X, \FF_q)$,
  is defined by functions (also known as cochains)
  which assign a value in $\FF_q$ to each cell of the complex.
In contrast, a cochain complex based on a sheaf,
  assigns a value in a vector space $\cF_\sigma$ to each cell $\sigma$,
  where the vector space $\cF_\sigma$ depends on the cell.
The vector spaces $\cF_\sigma$ are specified by $\cF_\tau$ for the $(t-1)$-cells $\tau \in X(t-1)$.
This additional structure $\{\cF_\tau\}_{\tau \in X(t-1)}$ is known as the local codes in Tanner code
  and plays a crucial role in constructing good qLDPC codes.

Given the sheaf $\cF$, we can define a cochain complex $\cC^{\bullet}(X, \cF)$,
  analogous to the standard cochain complex.
The corresponding quantum code is then obtained by
  choosing a particular dimension $\ell$
  where the X checks are placed at the $(\ell-1)$-cochains,
    the qubits are placed at the $\ell$-cochains,
    and the Z checks are placed at the $(\ell+1)$-cochains.

In algebraic topology, it is well-known that the cup product can be defined on the cochains $\cC^{\bullet}(X, \FF_q)$.
One of our main results is to extend the cup product to cochain complexes based on sheaves $\cC(X, \cF)$.
The cup product is defined on the cochains
\begin{equation}
  \cC^i(X, \cF_1) \times \cC^j(X, \cF_2)
  \xrightarrow{\cup} \cC^{i+j}(X, \cF_1 \odot \cF_2)
\end{equation}
and induces a corresponding cup product on cohomology classes
\begin{equation}
  H^i(X, \cF_1) \times H^j(X, \cF_2)
  \xrightarrow{\cup} H^{i+j}(X, \cF_1 \odot \cF_2).
\end{equation}
For each $(t-1)$-cell $\tau \in X(t-1)$,
  $(\cF_1 \odot \cF_2)_\tau \coloneqq
  \{c_1 \odot c_2: c_1 \in (\cF_1)_\tau, c_2 \in (\cF_2)_\tau\}$
  where $c_1 \odot c_2$ denotes the entrywise product of $c_1$ and $c_2$.

Using the cup product,
  we derive a sufficient condition for constructing qLDPC codes that support transversal CCZ gates.
In particular,
  if three sheaves $\cF_1, \cF_2, \cF_3$
  satisfy the condition that
  for every set of vectors $c_1 \in (\cF_1)_\tau, c_2 \in (\cF_2)_\tau, c_3 \in (\cF_3)_\tau$,
    we have $\sum_{i=1}^\Delta (c_1)_i (c_2)_i (c_3)_i = 0$,
  then this induces a transversal CCZ operation on the three quantum codes
  $\cC(X, \cF_1), \cC(X, \cF_2), \cC(X, \cF_3)$.

This general framework enables us to explore a wide family of codes,
  including the new qLDPC code with $\gamma \to 0$ studied in \cite{golowich2024quantum}.
Additionally, it clarifies the sufficient conditions to realizing transversal non-Clifford gates,
  which paves the way for future construction of qLDPC codes with improved parameters.

In addition to the result on the cup product
  we also study various aspects of sheaves over cellular complexes.
Specifically,
  we propose an alternative definition of sheaves that differs from \cite{curry2014sheaves}
  and prove an analogous statement of the Poincaré duality in the context of cellular complexes.

\subsection{Proof overview}

Our construction is inspired by various observations and previous works.
The first observation is that the triorthogonal code introduced in \cite{bravyi2012magic}
  heavily relies on the trilinear function
  $f(\alpha_1, \alpha_2, \alpha_3) = \sum_{i=1}^n (\alpha_1)_i (\alpha_2)_i (\alpha_3)_i$,
  where $\alpha_1, \alpha_2, \alpha_3$ are vectors in $\FF_2^n$.
This suggests that a degree $3$ function is necessary to access the non-Clifford gates at the third level of the Clifford hierarchy.
This also raises the possibility that other degree $3$ functions,
  such as more general trilinear functions,
  may be sufficient to capture the non-Clifford features.
This idea is formalized into the CCZ formalism discussed in \Cref{sec:CCZ-formalism}.
In particular,
  all we need is a trilinear function $f$ defined on the cohomology classes,
  which we mean by the condition
\begin{equation}\label{eq:intro}
  f(\z_1, \z_2, \z_3) = f(\z_1 + \b_1, \z_2 + \b_2, \z_3 + \b_3)
\end{equation}
for any $\z_i$ in the cocycle (corresponding to the logical operators)
  and any $\b_i$ in the coboundary (corresponding to the stabilizers).
This modification provides more flexibility to the construction of codes with transversal non-Clifford gates.

The second observation is that many known codes with non-Clifford operations
  can be linked to the triple intersection number of
  their corresponding topological spaces.
For example,
  the $[[8, 3, 2]]$ code \cite{campbell2016smallest}
  and 3D toric codes support transversal CCZ gates \cite{vasmer2019three}.
This is related to the fact that the 3D torus
  has three homology clases $H_2(\mathbb{T}^3, \ZZ/2\ZZ) = (\ZZ/2\ZZ)^3$
  with a nontrivial triple intersection.
Furthermore,
  the topological nature of the triple intersection number
  implies that it is defined on the homology classes.
This observation suggests an approach to study quantum code with a topological nature
  and define $f$ as the triple intersection.
This strategy has also been observed recently
  and applied to 3D manifolds in \cite{zhu2023non}.

This brings us to the final observation that the good qLDPC codes
  in the recent breakthrough \cite{panteleev2021asymptotically,leverrier2022quantum,dinur2022good}
  do have a topological interpretation,
  including their higher dimensional analogue \cite{dinur2024expansion}.
This topological perspective is inspired by a series of work on geometrically local codes
  \cite{portnoy2023local,lin2023geometrically,williamson2023layer,li2024transform}
  which implicitly uses the topological nature of the qLDPC codes.
For qLDPC codes constructed from 3D cubical complexes,
  the X logical operators can be interpreted as surface configurations
  on the cubical complex.
This geometric interpretation allows us to define the triple intersection number,
  which leads to the transversal CCZ gates on these codes.
Together, these observations provide the foundation for our main result regarding the cup product on chain complexes induced from sheaves.

\subsection{Further directions}

\paragraph{Quantum LDPC codes with transversal non-Clifford gates and nearly linear distance.}
An immediate question is to improve the code parameters
  obtained in the companion work \cite{golowich2024quantum}.
It would be desireable to obtain qLDPC codes with linear distance (or up to polylogs).
Achieving such a result should also lead to quantum locally testable codes (qLTCs) with transversal non-Clifford gates
  with inverse polylog soundness,
  by extending the arguments from \cite{dinur2024expansion}.

\paragraph{Fully-fledged fault-tolerant scheme.}
With the improved qLDPC codes and qLTCs,
  the next goal is to construct a complete fault-tolerant quantum computing scheme based on such codes.
A key question here is to develop an efficient
  code-switching protocols between these codes.
This probably relies on the relation between the cellular complexes these codes are based on.
In particular, the structural property
  that a cubical complex contains lower-dimensional cubical complexes,
  may facilitate code-switching.

\paragraph{Quantum PCP conjecture.}
The fault-tolerant scheme based on quantum LDPC codes with transversal gates
  may also contribute to progress on the quantum PCP conjecture \cite{aharonov2013guest}.
It has been pointed in \cite{anshu2024circuit}
  that an adversarially robust fault-tolerant scheme can be used
  to construct local Hamiltonians for which
  it is QMA-hard to decide whether the energy density is $< a$ or $> a + 1/\poly\log n$.
If the scheme using qLDPC codes with transversal gates
  proves to be adversarially robust,
  it would be a major step toward addressing the qPCP conjecture.

\section{Preliminary}

\subsection{Chain complex}
\label{sec:chain-complex}

We review some definitions for chain complexes,
  which are helpful for explaining the technical parts of the paper.
In our context, all vector spaces are equipped with a chosen basis.
Additionally, we place significant emphasis on the LDPC condition,
  which is essentially about sparsity of matrices.
A matrix is sparse if each row and column
  contains at most a constant number of nonzero entries.
Note that the definition requires a chosen basis on the vector spaces.

\begin{definition}[Chain complex]
  A chain complex $\cC$ is collection of vector spaces over $\FF_q$, $\cC^i$,
    and linear maps $\delta^i: \cC^i \to \cC^{i+1}$,
    such that $\delta^{i+1} \delta^i = 0$.

  A chain complex with basis $\cC$,
    additionally have a chosen basis for every vector space $\cC^i$.

  A chain complex with basis $\cC$ is sparse
    if every coboundary $\delta^i$ is sparse under the chosen basis.
\end{definition}

\begin{definition}[Chain map]
  A chain map $f$ from $\cC$ to $\cD$
    is a collection of maps $f^i: \cC^i \to \cD^i$,
    such that $\delta_\cD^i f^i = f^{i+1} \delta_\cC^i$.

  A chain map $f$ is sparse if every $f^i$ is sparse under the chosen basis.
\end{definition}

\begin{definition}[Chain homotopy]
  Given two chain maps $f, g$ from $\cC$ to $\cD$,
    a chain homotopy from $f$ to $g$ is a collection of maps $h^i: \cC^i \to \cD^{i-1}$
    such that $f^i - g^i = \delta_{\cD}^{i-1} h^i + h^{i+1} \delta_{\cC}^i$.

  A chain homotopy $h$ is sparse if every $h^i$ is sparse under the chosen basis.

  Two chain complexes $\cC, \cD$ are homotopy equivalent
    if there exists chain maps $f$ from $\cC$ to $\cD$
      and $g$ from $\cD$ to $\cC$
    such that $g \circ f$ is homotopic to $\id_{\cC}$
      and $f \circ g$ is homotopic to $\id_{\cD}$.

  Such homotopy equivalence is sparse
    if the chain homotopies are sparse.
\end{definition}

From now on, we only consider the sparse version of every definition above.
Therefore, we often drop the word sparse.

Given a chain complex, we can change the direction of the linear maps
  and form the dual chain complex consists of the boundary operators.
Let $\cC_i$ be the dual vector space $\cC^i \to \FF_q$.
In our context,
  since each vector space has a chosen basis,
  we can identify $\cC_i \cong \cC^i$.
Furthermore, we can define the associate boundary operators
  $\partial_i: \cC_i \rightarrow \cC_{i-1}$ as the matrix transpose of $\delta^{i-1}$,
  $\partial_i \coloneqq (\delta^{i-1})^T$.
The boundary operators automatically satisfy
\begin{equation*}
  \partial_{i-1} \partial_i  = 0
\end{equation*}
which again forms a chain complex.
For most of the discussion, we will focus on cochains rather than chains.

We introduce some standard terminologies.
The elements in $\cC^i$ ($\cC_i$) are called the $i$-cochains ($i$-chains).
The elements of the kernel of the coboundary (boundary) operators are called cocycles (cycles)
\begin{equation}
  Z^i \coloneqq \ker \delta^i = \{\a \in \cC^i : \delta^i \a = 0\},\qquad
  Z_i \coloneqq \ker \partial_i = \{\a \in \cC_i : \partial_i \a = 0\}.
\end{equation}
The elements of the image of the coboundary (boundary) operators are called coboundaries (boundaries)
\begin{equation}
  B^i \coloneqq \im \delta^{i-1} = \{\delta^{i-1} \a : \a \in \cC^{i-1}\},\qquad
  B_i \coloneqq \im \partial_{i+1} = \{\partial_{i+1} \a : \a \in \cC_{i+1}\}.
\end{equation}
Because $\delta^i \delta^{i-1} = 0$ ($\partial_i \partial_{i+1} = 0$), it follows that $B^i \subseteq Z^i$ ($B_i \subseteq Z_i$).
The cohomology (homology) group is the quotient of the coboundaries by the cocycles (boundaries by the cycles)
\begin{equation}
  H^i \coloneqq Z^i / B^i,\qquad
  H_i \coloneqq Z_i / B_i.
\end{equation}
The elements are known as the cohomology classes
  which will be denoted as the equivalence class of a coboundary
  $[\zeta] = \{\zeta + \beta: \beta \in B^i\}$
  for $\zeta \in Z^i$.

\subsection{Quantum operators over finite fields}
\label{sec:qudit}

We briefly review qudits over finite fields \cite{ashikhmin2001nonbinary}.
Let $q = p^r$, be a power of primes.
The basis of qudits over $\FF_q$ is generated by $\{|x\>: x \in \FF_q\}$.
The Pauli operators over $\FF_q$ are
\begin{equation}
  X^a |x\> = |x+a\>, \qquad Z^a |x\> = e^{\frac{2\pi i}{p} \tr_{\FF_q / \FF_p}(ax)} |x\>
\end{equation}
for $a \in \FF_q$.
Notice that $Z^b X^a$ and $X^a Z^b$ differ by a phase factor
\begin{equation}
  Z^b X^a = e^{\frac{2\pi i}{p} \tr_{\FF_q / \FF_p}(ab)} X^a Z^b.
\end{equation}

With the Pauli operators defined,
  the Clifford hierarchy can be generalized to the setting of qudits.
Let $C_1$ be the Pauli group generated by the Pauli operators.
We define the Clifford hierarchy iteratively,
  where the $i$-th level $C_i$ is defined as
\begin{equation}
  C_i = \{U: U C_1 U^\dagger \subseteq C_{i-1}\}.
\end{equation}
The second level $C_2$ are also known as the Clifford group,
  and its elements are known as the Clifford gates.
The hierarchy characterize some form of computing power.
In particular, the gates in $C_2$ do not form a universal gate set
  and can be simulated efficiently using classical computers.
In contrast, the gates in $C_3$ form a universal gate set.
These gates outside of $C_2$ are known as the non-Clifford gates
  and are essential for building universal quantum computers.

We illustrate some of the gates in the hierarchy.
At the second level $C_2$, for $a \in \FF_q$, we have the following operators
\begin{align}
  M^a |x\> &= |ax\>, \\
  CZ^a |x,y\> &= e^{\frac{2\pi i}{p} \tr_{\FF_q / \FF_p}(axy)} |x, y\>, \\
  CNOT^a |x,y\> &= |x, ax+y\>.
\end{align}
At the third level $C_3$, for $a \in \FF_q$, we have the following operators
\begin{align}
  T^a |x\> &= e^{\frac{2\pi i}{p} \tr_{\FF_q / \FF_p}(a x^3)} |x\>, \\
  CCZ^a |x,y,z\> &= e^{\frac{2\pi i}{p} \tr_{\FF_q / \FF_p}(a xyz)} |x, y, z\>.
\end{align}
In particular, since $M^a$ is a Clifford gate,
  in the context of magic state distillation,
  $CCZ^a$ and $CCZ^1$ are of the same resources.
This follows from the relation
\begin{equation}
  CCZ^a = M^{a^{-1}} CCZ^1 M^a.
\end{equation}

It is common to perform alphabet reduction
  which reduces qudits over $\FF_q$ into $r$ qudits over $\FF_p$.
This allows us to obtain codes over qubits ($p=2$) while working with large finite fields ($q=2^r$).
By viewing $\FF_q$ as a vector space over $\FF_p$, $\FF_q \cong (\FF_p)^r$,
  we can express $|x\>$ as $|x_1\> \otimes |x_2\> \otimes ... \otimes |x_r\>$,
  where $x_1, ..., x_r \in \FF_p$.
One can verify that the Clifford hierarchies over $\FF_q$ and $\FF_p$ agree
  after alphabet reduction.
Conceptually, this means that the computation model only depends on the characteristic of the field.
More concretely,
  there are several methods to perform alphabet reduction
  while preserving the properties of the stabilizer codes \cite{golowich2024asymptotically,nguyen2024good}.
Therefore, there is no need to be intimidated by the large field size $\FF_q$.

Notice that several notations in this section
  overlap with those used in the chain complex.
From now on, we will only use $CCZ$,
  while $X$ and $Z$ will mainly refer to concepts in the context of the chain complex.

\subsection{Quantum CSS codes}

A quantum CSS code of length $n$ over $\FF_q$ is specified by a three term chain complex
\begin{equation}
  Q: \FF_q^{m_x} \xrightarrow{\delta^0} \FF_q^n \xrightarrow{\delta^1} \FF_q^{m_z}
\end{equation}
where $m_x$ is the number of X checks,
  $n$ is the number of qudits,
  and $m_z$ is the number of Z checks.
The coboundary maps $\delta^0$ and $\delta^1$ specify how the checks act on the qudits.

It is understood that the X and Z logical operators correspond to cocycles $Z^1$ and cycles $Z_1$, respectively.
Additionally, the X and Z stabilizers correspond to coboundaries $B^1$ and boundaries $B_1$, respectively.

The dimension $k = \dim H^1$ is the number of logical qubits.
The distance is $d = \min(d^1, d_1)$ where
\begin{equation}
  d^1 = \min_{\alpha \in Z^1 \backslash B^1} |\alpha|,\qquad
  d_1 = \min_{\alpha \in Z_1 \backslash B_1} |\alpha|.
\end{equation}
These are the minimal weights of the nontrivial $X$ and $Z$ logical operators.

We say the quantum code is a low-density parity-check (LDPC) code if the chain complex is sparse,
  i.e. each check interacts with a bounded number of qudits
  and each qudit interacts with a bounded number of checks.

To simplify notation in later sections,
  where many indices are needed to label variables,
  we will sometimes drop the indices and use $Z$ for $Z^1$ and $B$ for $B^1$
  as the discussion mostly focuses on the properties of cochains.

\subsection{Cubical complexes and cubical codes}
\label{sec:cubical-code}

The cubical complexes are formed by vertices, edges, squares, cubes, and their high dimensional counterparts.
Their first application in error correcting codes
  was introduced in \cite{dinur2024expansion}
  to construct quantum locally testable codes.
Prior to this, their 1D and 2D counterparts had been applied
  in \cite{sipser1996expander,tanner1981recursive}
  and \cite{panteleev2022asymptotically,dinur2021locally}.
Our discussion will focus on labeled cubical complexes
  because it is more straightforward to apply local codes.

\subsubsection{Labeled cubical complexes}

A $t$-dimensional labeled cubical complex $X(V; A_1, A_2, ..., A_t)$
  can be constructed from a set $V$ of size $N = |V|$
  and subsets of permutations $A_1, ..., A_t \subset \Aut(V)$
  of size $\Delta = |A_i|$.
These permutations should pairwise commute,
  i.e. $a_i a_j = a_j a_i$ for all $a_i \in A_i$ and $a_j \in A_j$.
The complex is also denoted as $X(V; \{A_i\})$.

Each $k$-dimensional cell ($0 \le k \le t$)
  is associated with a type,
  which is a subset $S \subseteq \{1, ..., t\}$ of size $|S| = k$.
Denote $\ol{S}$ as the complement of $S$ in $\{1, ..., t\}$.
A cell $\sigma$ of type $S$ is specified by
\begin{equation}
  \sigma = (v; (a_j)_{j \in S}, (b_j)_{j \in \ol{S}}),
\end{equation}
where $v \in V$, $a_j \in A_j$ for each $j \in S$, and $b_j \in \{0, 1\}$ for each $j \in \ol{S}$.
Let $X(S)$ be the faces of type $S$.
Then, $|X(S)| = |V|\Delta^k 2^{t-k}$.
The type of a cell $\sigma$ is denoted as $\type(\sigma)$.

Geometrically, a $k$-cell of type $S$, $\sigma = (v; a, b)$
  contains the following $2^k$ vertices
  \begin{equation}
    \Big\{\Big(\prod_{j: b'_j = 1} a_j \cdot v; b' || b\Big) : b' \in \{0, 1\}^S \Big\},
  \end{equation}
  where $b' || b$ denotes string concatenation with re-ordering of the indices in the natural way.
We use the notation $a_j \cdot v$ to denote the element $a_j(v) \in V$
  for $a_j \in A_j$ is a permutation of $V$.
Note that the order of the product in $\prod_{j} a_j \cdot v$ does not matter
  since the permutations $a_j$ from different sets $A_j$ pairwise commute.

A family of labeled cubical complexes to keep in mind
  is to take $V$ to be an abelian group $G$,
  and $A_i$ to be a collection of left actions
    $A_i = \{(g \mapsto \tilde a_i g): \tilde a_i \in \tilde A_i\}$
    for a generating subset $\tilde A_i \subset G$.
Since the group is abelian,
  the elements from different sets pairwise commute.
Because of this example, from now on,
  we will use $G$ instead of $V$ to denote the vertex set.

When $t = 2$,
  we can take $G$ to be a non-abelian group,
  $A_1$ to be a collection of left actions
  $A_1 = \{(g \mapsto \tilde a_1 g): \tilde a_1 \in \tilde A_1\}$,
  and $A_2$ to be a collection of right actions
  $A_2 = \{(g \mapsto g \tilde a_2): \tilde a_2 \in \tilde A_2\}$.
The permutations $a_1 \in A_1$ and $a_2 \in A_2$ commutes,
  because left and right actions commute.
This is known as the left-right Cayley complex \cite{dinur2021locally}
  which is an essential ingredient to the recent breakthroughs in
  classical locally testable codes and quantum low-density parity check codes \cite{dinur2021locally,panteleev2021asymptotically}.

Another family of labeled cubical complexes comes from a Cartesian product of three bipartite graphs.
These are used for constructing qLDPC codes with $\gamma \to 0$
  in the companion paper \cite{golowich2024quantum}.

We introduce some additional notations.
Let $X(i)$ be the set of $i$-cells.
Given two cells $\sigma, \tau$,
  we write $\sigma \preceq \tau$ if $\sigma$ is a subcell of $\tau$,
  and write $\sigma \precdot \tau$ if $\sigma$ is an `immediate' subcell of $\tau$
    where their dimensions differ by $1$.

$X_{\ge \sigma}$ denotes the sub-complex consists of all cells above $\sigma$
\begin{equation}
  X_{\ge \sigma} = \{\tau \in X: \tau \succeq \sigma\}.
\end{equation}
$X_{\le \sigma}$ denotes the sub-complex consists of all cells below $\sigma$
\begin{equation}
  X_{\le \sigma} = \{\tau \in X: \tau \preceq \sigma\}.
\end{equation}
We define $X_{\ge \sigma}(k) = X_{\ge \sigma} \cap X(k)$
  and $X_{\le \sigma}(k) = X_{\le \sigma} \cap X(k)$.

\subsubsection{Chain complexes from cubical complexes}

To generate a chain complex,
  we attach each cell with local coefficients similar to a Tanner code construction.
Let $C_1, ..., C_t \subseteq \FF_q^{\Delta}$
  be classical codes of length $\Delta$ with dimension $m_i$
  and generating matrix $h_i^T: \FF_q^{m_i} \to \FF_q^\Delta$.
The local coefficient space associated with a $k$-cell $\sigma$ of type $S$
  is $\cF_\sigma = \bigotimes_{i \in \ol{S}} \FF_q^{m_i}$.

The local coefficient can be identified with a codeword in the tensor code
  $\bigotimes_{i \in \ol{S}} C_i$.
Given $c \in \cF_\sigma$,
  it is clear that $(\bigotimes_{i \in \ol{S}} h_i^T) c$ is a codeword in $\bigotimes_{i \in \ol{S}} C_i \subseteq (\FF_q^{\Delta})^{\otimes \ol{S}}$.
The codeword can be further interpreted as a function
  that assigns every $t$-cells above $\sigma$ with a value in $\FF_q$,
  $(\bigotimes_{i \in \ol{S}} h_i^T) c: X_{\ge \sigma}(t) \to \FF_q$.
Since each $t$-cell above $\sigma$ is obtained by extending the directions in $\ol{S}$
  and each direction has $|A_i| = \Delta$ labels to choose from,
  we can identify $\Delta^{\ol{S}}$ with $X_{\ge \sigma}(t)$
  and obtain $(\FF_q^{\Delta})^{\otimes \ol{S}} \cong \FF_q^{X_{\ge \sigma}(t)} \cong (X_{\ge \sigma}(t) \to \FF_q)$.
This induces a map from $\cF_\sigma$ to $\FF_q^{X_{\ge \sigma}}$
  and we define $\cG_\sigma \subseteq \FF_q^{X_{\ge \sigma}(t)}$
  as the image.
It is clear that $\cG_\sigma \cong \cF_\sigma$.
These $\cG_\sigma$ will be helpful when discussing the generalization to sheaf code in \Cref{sec:sheaf-code}
  and defining the co-restriction map below.

We first define the vector spaces for the desired chain complex.
The vector space $\cC^i$ is the space of functions $\alpha$,
  that assigns a value $\alpha(\sigma) \in \cF_\sigma$
  for each $i$-cell $\sigma$.
In particular,
\begin{equation}
  \cC^i = \bigoplus_{\sigma \in X(i)} \cF_\sigma.
\end{equation}

To define the coboundary maps,
  we first describe the co-restriction maps.
For every pair of cells $\sigma \preceq \tau$,
  we have an embedding $X_{\ge \tau}(t) \subseteq X_{\ge \sigma}(t)$.
This induces the map $\cG_\sigma \to \cG_\tau$,
  which maps $c: X_{\ge \sigma}(t) \to \FF_q$
  to $c|_{X_{\ge \tau}(t)}: X_{\ge \tau}(t) \to \FF_q$.
Because $\cG_\sigma \cong \cF_\sigma$,
  this induces a map from $\cF_\sigma \to \cF_\tau$,
  which we denote as $\cores_{\sigma,\tau}$.

More explicitly, $\cores$ can be defined by mapping through $h_j^T$
  among certain directions and restrict to certain slices.
Given a pair of cells $\sigma \precdot \tau$
  with $\type(\tau) \backslash \type(\sigma) = \{j\}$,
\begin{equation}
  \cores_{\sigma,\tau}(c) = \Big((I_{-j} \otimes h_j^T) c\Big)[..., a_i, ...],
\end{equation}
  where $a_j \in A_j$ is the label for the $j$-th direction of $\tau$.
For a general pair of cells $\sigma \preceq \tau$,
  we again apply $h_j^T$ and restrict to $a_j \in A_j$
  for all $j \in \type(\tau) \backslash \type(\sigma)$.
One can verify that this agrees with the previous definition of $\cores$.

The coboundary maps $\delta^i: \cC^i \to \cC^{i+1}$
  is defined by
\begin{equation}
  \delta^i \a (\tau) = \sum_{\sigma \precdot \tau} \cores_{\sigma,\tau}(\a(\sigma)).
\end{equation}
(In general, there will be signs, $\sum_{\sigma \precdot \tau} (-)^{(...)} \cores_{\sigma,\tau}(\a(\sigma))$.
  We evade this technical discussion by studying finite fields with characteristic $2$
  where $-1 = 1$.)

We can verify that $\delta \circ \delta = 0$, where we have
\begin{equation}
  \delta^{i+1} \delta^i \a(\pi) = \sum_{\sigma \precdot \tau \precdot \pi} \cores_{\sigma,\pi}(\a(\sigma)) = 0.
\end{equation}
The last equality holds
  because for each pair $\sigma \in X(i), \pi \in X(i+2)$,
  there exist two $\tau \in X(i+1)$ such that $\sigma \precdot \tau \precdot \pi$.

We denote the chain complex as $\cC(X, \{C_j\}_{j=1}^t)$
  or $\cC(G, \{A_j\}_{j=1}^t, \{C_j\}_{j=1}^t)$.

\subsubsection{Cubical quantum codes}

Recall that a quantum CSS code is equivalent to a three-term chain complex.
The previous section provides a chain complex
  $\cC(G, \{A_i\}, \{C_i\})$.
So the only additional data is an integer $\ell$
  which specifies which of the three terms to select.
This leads to the quantum code $Q(G, \{A_i\}, \{C_i\}, \ell)$
  defined by
\begin{equation}
  \cC^{\ell-1}(G, \{A_i\}, \{C_i\}) \to \cC^{\ell}(G, \{A_i\}, \{C_i\}) \to \cC^{\ell+1}(G, \{A_i\}, \{C_i\}).
\end{equation}
In particular,
  the X checks are placed on the $(\ell-1)$-cells,
  the qubits are placed on the $\ell$-cells,
  and the Z checks are placed on the $(\ell+1)$-cells,

Later, we will mainly discuss quantum codes based on
  2D square complexes and 3D cubical complexes
  with $\ell = 1$.

\section{CCZ formalism}
\label{sec:CCZ-formalism}

This section introduces the CCZ formalism,
  where the goal is to establish a condition on quantum codes
  that implies a transversal operation of CCZ gate.
This framework should be compared to the triorthogonal condition introduced in \cite{bravyi2012magic},
  which specifies a condition on quantum codes
  that implies a transversal operation for the T gate.
Our formalism can be generalized to C$^{r-1}$Z gates or subsystem codes,
  although we will not discuss them explicitly in this work.

\subsection{Definition of CCZ codes}

In the CCZ formalism,
  the data consists of
  three quantum codes $Q_1, Q_2, Q_3$ of length $n_1, n_2, n_3$, respectively,
  and a trilinear map $f: \FF_q^{n_1} \times \FF_q^{n_2} \times \FF_q^{n_3} \to \FF_q$.
We say a function is trilinear if
  the function is linear when any of the two inputs are fixed.
For example, $f(\a_1,\a_2,\a_3) + f(\a_1',\a_2,\a_3) = f(\a_1+\a_1',\a_2,\a_3)$
  and $f(a \a_1,\a_2,\a_3) = a f(\a_1,\a_2,\a_3)$ for $a \in \FF_q$.
The trilinear map $f$ has to satisfy an additional condition which we will explain shortly.

We associate the trilinear map $f$ with an unitary operator $\CCZ^f$ defined by
\begin{equation}\label{eq:def-CCZ-f}
  \CCZ^f |\a_1,\a_2,\a_3\> = e^{\frac{2\pi i}{p}\tr_{\FF_q/\FF_p}(f(\a_1,\a_2,\a_3))} |\a_1,\a_2,\a_3\>
\end{equation}
where $\a_i \in \FF_q^{n_i}$
  and $|\a_1,\a_2,\a_3\>$ is the state on $n_1 + n_2 + n_3$ qudits
  labeled by strings $\a_1, \a_2, \a_3$ in the Z basis.
(Note that strings in the Z basis correspond to Pauli X operators.)
This is analogous to the definition of a CCZ gate.
In fact, we will see that $\CCZ^f$ can be implemented by applying multiple CCZ gates.

To ensure that $\CCZ^f$ maps codewords to codewords,
  we need to impose a condition on $f$.
The codewords in $Q_1 \otimes Q_2 \otimes Q_3$ are spanned by
\begin{equation}
  |[\z_1],[\z_2],[\z_3]\>
  \coloneqq \sum_{\b_1 \in B_1} \sum_{\b_2 \in B_2}\sum_{\b_3 \in B_3}
  |\z_1 + \b_1, \z_2 + \b_2, \z_3 + \b_3\>
\end{equation}
for $\z_i \in Z_i$.
Because $\CCZ^f$ only changes the phase under the Z basis,
  $\CCZ^f |[\z_1],[\z_2],[\z_3]\>$ is a codeword iff
\begin{equation}
  \CCZ^f |[\z_1],[\z_2],[\z_3]\> \,\propto\, |[\z_1],[\z_2],[\z_3]\>
\end{equation}
which is equivalent to
\begin{equation}\label{eq:condition-with-trace}
  \tr_{\FF_q/\FF_p}(f(\z_1,\z_2,\z_3)) = \tr_{\FF_q/\FF_p}(f(\z_1 + \b_1, \z_2 + \b_2, \z_3 + \b_3))
\end{equation}
for all $\z_i \in Z_i$ and $\b_i \in B_i$.
Therefore, we require $f$ to satisfy
\begin{equation}\label{eq:condition-without-trace}
  f(\z_1,\z_2,\z_3) = f(\z_1 + \b_1, \z_2 + \b_2, \z_3 + \b_3)
\end{equation}
which guarantees that $\CCZ^f$ maps codewords to codewords.
This condition is appealing as it is saying that
  $f$ is uniquely defined on the cohomology classes.
This is how we will refer to this condition.

\begin{remark}
  The condition we imposed \eqref{eq:condition-without-trace}
  may seem stronger than what is necessary \eqref{eq:condition-with-trace}.
  However, they are in fact equivalent.
  For all $a \in \FF_q$,
    because $a \z_1 \in Z_1$ and $a \b_1 \in B_1$,
    from \eqref{eq:condition-with-trace} we have
    \begin{equation}
      \tr_{\FF_q/\FF_p}(f(a \z_1,\z_2,\z_3)) = \tr_{\FF_q/\FF_p}(f(a \z_1 + a \b_1, \z_2 + \b_2, \z_3 + \b_3)).
    \end{equation}
  Because $f$ is trilinear, we have
  \begin{equation}
    \tr_{\FF_q/\FF_p}(a f(\z_1,\z_2,\z_3)) = \tr_{\FF_q/\FF_p}(a f(\z_1 + \b_1, \z_2 + \b_2, \z_3 + \b_3))
  \end{equation}
  for all $a \in \FF_q$.

  Since finite extensions of a finite field is separable,
    we have $\FF_q/\FF_p$ is separable.
  Furthermore, for separable extensions $\FF_q/\FF_p$,
    the bilinear trace form $\tr_{\FF_q / \FF_p}(ax)$ is nondegenerate,
    meaning that if $\tr_{\FF_q / \FF_p}(ax) = 0$ for all $a \in \FF_q$,
    then $x = 0$.
  Thus, $f(\z_1,\z_2,\z_3) = f(\z_1 + \b_1, \z_2 + \b_2, \z_3 + \b_3)$
  which is exactly \eqref{eq:condition-without-trace}.
\end{remark}

This leads to the definition of a CCZ code.
\begin{definition}
  A CCZ code consists of
  three quantum codes $Q_1, Q_2, Q_3$
  and a trilinear map $f: \FF_q^{n_1} \times \FF_q^{n_2} \times \FF_q^{n_3} \to \FF_q$
  such that
  \begin{equation}\label{eq:condition-without-trace-main}
    f(\z_1,\z_2,\z_3) = f(\z_1 + \b_1, \z_2 + \b_2, \z_3 + \b_3)
  \end{equation}
  for all $\z_i \in Z_i$ and $\b_i \in B_i$.
\end{definition}
Generally, $n_1, n_2, n_3$ are of the same order.

\begin{example}[Triorthogonal code]
  Three copies of the same triorthogonal code \cite{bravyi2012magic} form a CCZ code.
  Let $Q$ be a triorthogonal code of length $n$.
  Let $\{\b^a\}$ be the vectors correspond to stabilzer generators
    and let $\{\z^a\}$ be vectors corresponds to logical generators.
  Recall that a triorthogonal code is defined over $\FF_2$
    and satisfies
  \begin{enumerate}
    \item $\sum_{i=1}^n \b^a_i = 0$,
    \item $\sum_{i=1}^n \z^a_i = 1$,
    \item $\sum_{i=1}^n \a^a_i \a^b_i = 0$ for any distinct $\a^a, \a^b \in \{\b^a\} \cup \{\z^a\}$,
    \item $\sum_{i=1}^n \a^a_i \a^b_i \a^c_i = 0$ for any distinct $\a^a, \a^b, \a^c \in \{\b^a\} \cup \{\z^a\}$.
  \end{enumerate}

  Set $f(\a_1, \a_2, \a_3) = \sum_{i=1}^n \a_{1,i} \a_{2,i} \a_{3,i}$,
    where $\a_{1,i}$ is the $i$-th entry of $\a_1$.
  It is straightforward to check that the conditions for triorthogonal codes
    imply the condition for CCZ codes, \Cref{eq:condition-without-trace-main}.
\end{example}

\begin{example}[3D surface code] \label{exam:3D-surface}
  Three 3D surface codes form a CCZ code.
  Since similar facts have been observed in \cite{vasmer2019three,zhu2023non},
    we will not construct $f$ explicitly.
  We will instead make some high level remarks.

  The key idea is that the X-logical operators of the 3D surface codes
    can be viewed as 1-forms.
  A natural way to form a trilinear operator among three 1-forms is through the cup product.
  Schematically, we can write
  $f(\alpha_1, \alpha_2, \alpha_3) = \int_M \alpha_1 \cup \alpha_2 \cup \alpha_3$.
  It is well-known in algebraic topology that cup product is well-defined on the cohomology classes.
  Therefore, the construction automatically satisfy the condition for CCZ codes.

  An alternative perspective, which may be more visually appealing,
    is to note that 1-forms are Poincaré dual to 2D surfaces.
  In this dual perspective, $f$ becomes the intersection number of the three 2D surfaces.
  This can be visualized in \Cref{fig:3D-surface}.
  In particular, the fact that there is ``one'' nontrivial triple intersection
    corresponds to having ``one'' logical CCZ gate.

  In this example, we see that the reason $f$ is defined on the cohomology classes
    is largely induced from the topological nature of the construction.
  In particular, the intersection number is a topological invariant.
  This is the main idea that allows us to construct transversal CCZ gates on qLDPC codes.
  Later, we will first provide a topological interpretation of the qLDPC codes,
    which naturally induces the construction of $f$ through the intersection number.

  \begin{figure}[H]
    \centering
    \includegraphics[width=0.3\linewidth]{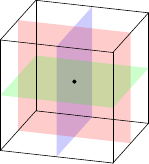}
    \caption{The X-logical operators of three 3D surface codes,
              each oriented along a different axis.
            These logical operators are represented by the corresponding surfaces.
            A crucial feature of these surfaces is that their triple intersection number is always $1$.}
    \label{fig:3D-surface}
  \end{figure}
\end{example}

\begin{remark}
  Codes with transversal T gates appear to be related to triple self-intersection.
  An example with nontrivial triple self-intersection
    is the 3D real projective space $\mathbb{RP}^3$
    which has a homology $H_2(\mathbb{RP}^3, \ZZ/2\ZZ) = \ZZ/2\ZZ$ with a nontrivial triple self-intersection.
\end{remark}

\subsection{Extra code parameters of CCZ codes: $n_{\CCZ}, k_{\CCZ}$}

Besides the traditional code parameters, such as $n, k, d$,
  a CCZ code has additional parameters $n_{\CCZ}$, $k_{\CCZ}$, which are dependent on $f$.
Intuitively, $n_{\CCZ}$ is the number of CCZ gates needed to implement $\CCZ^f$,
  and $k_{\CCZ}$ is the number of logical CCZ gates applied by acting $\CCZ^f$.
In the context of magic state distillation,
  this means, by consuming $n_{\CCZ}$ noisy CCZ states,
  the error correcting code can output $k_{\CCZ}$ cleaner CCZ states.
Therefore, the corresponding magic state distillation protocol
  has overhead $\log^\gamma(1/\epsilon)$ for $\gamma = \log(n_{\CCZ}/k_{\CCZ}) / \log(d)$, where $\epsilon$ is the target error rate.

We first discuss $n_{\CCZ}$.
Notice that the operator $\CCZ^f$ defined in \eqref{eq:def-CCZ-f}
  can be rewritten as
\begin{equation}
  \CCZ^f
  = \prod_{j_1 \in [n_1], j_2 \in [n_2], j_3 \in [n_3]} \CCZ^{f(e_{j_1},e_{j_2},e_{j_3})}_{j_1, j_2, j_3}
  = \prod_{j_1 \in [n_1], j_2 \in [n_2], j_3 \in [n_3], f(e_{j_1},e_{j_2},e_{j_3}) \ne 0} \CCZ^{f(e_{j_1},e_{j_2},e_{j_3})}_{j_1, j_2, j_3}
\end{equation}
where $\CCZ^{f(e_{j_1},e_{j_2},e_{j_3})}_{j_1, j_2, j_3}$
  is the $\CCZ^{f(e_{j_1},e_{j_2},e_{j_3})}$ gate acting on the $j_1, j_2, j_3$-th qudits in codes $Q_1, Q_2, Q_3$, respectively.
Recall that, $\CCZ^a$ gates for different $a \ne 0$
  are equivalent resources up to Clifford operations, as discussed in \Cref{sec:qudit}.

One may notice that this application of CCZ gates
  is not strictly transversal.
There are two perspective to address this concern.
First, for the purpose of constructing fault-tolerant circuits,
  the encoded unitary does not need to be strictly transversal.
In fact, the same fault-tolerant analysis holds as long as the circuit has constant depth,
  which is primarily the setting we consider.
Since ``transversal'' is a well-established term,
  we will slightly abuse its meaning.
Similarly, constant depth circuits are also acceptable
  for magic state distillation.

Alternatively, constant depth applications of CCZ gates
  can be transformed into strictly transversal operation
  by enlarging the quantum code, \cite[Lemma 2.45]{golowich2024quantum}.
Conceptually, this can be compared to the weight reduction techniques,
  which also implies a certain equivalence between different locality constants.

The above discussion leads to the following definition.
\begin{definition}
  Given a CCZ code, we define
  \begin{equation}\label{eq:def-n-CCZ}
    n_{\CCZ} \coloneqq |\{(j_1, j_2, j_3): f(e_{j_1},e_{j_2},e_{j_3}) \ne 0\}|
  \end{equation}
  where $e_j$ is the standard basis vector with all entries equal to $0$ except for the $j$-th entry which is $1$.
\end{definition}

It is clear that $n_{\CCZ} \le n_1 n_2 n_3$.
However, in the context of fault tolerance,
  we want $f$ to be sparse,
  meaning that each qubit is only involved in a bounded number of CCZ gates.
This sparsity condition is analogous to the LDPC condition.
Since each $j_1$ only appears a bounded number of times on the RHS of \eqref{eq:def-n-CCZ},
  in this context, we have $n_{\CCZ} = O(n_1)$.
In fact, we often have $n_{\CCZ} = \Theta(n_1) = \Theta(n_2) = \Theta(n_3)$.

\begin{remark}
  In magic state distillation with perfect Clifford gates,
    $n_{\CCZ}$ is a key code parameter
    as it quantifies the number of resource state consumed.
  However, in the context of transversal gates for building fault-tolerant circuits,
    it is crucial for the implementation of $\CCZ^f$ to have a low depth.
  That means it is important to bound the maximal number of CCZ gates that a qubit participates in,
    which leads to the definition of the code parameter $w_{\CCZ} = \max(w_1, w_2, w_3)$,
    where
  \begin{equation}
    w_1 = \max_{j_1 \in [n_1]} |\{(j_2, j_3) \in [n_2] \times [n_3]: f(e_{j_1},e_{j_2},e_{j_3}) \ne 0\}|,
  \end{equation}
  and $w_2$ and $w_3$ are defined similarly.
\end{remark}

We now discuss $k_{\CCZ}$, which is the number of logical CCZ gates applied by $\CCZ^f$.
Let $k_{\CCZ}$ be the largest integer
  such that there exist $k_{\CCZ}$ X logical operators in each of the codes $Q_1, Q_2, Q_3$,
  denoted as $\z_{1,j}, \z_{2,j}, \z_{3,j}$ for $j \in [k_{\CCZ}]$,
  with the property that
  \begin{equation}
    f(\z_{1,j_1}, \z_{2,j_2}, \z_{3,j_3}) = 1_{j_1 = j_2 = j_3}
  \end{equation}
  for $j_1, j_2, j_3 \in [k_{\CCZ}]$,
  where $1_{j_1 = j_2 = j_3}$ is equal to $0$ except when $j_1 = j_2 = j_3$ which is $1$.
Once we found those X logical operators,
  we treat them as the data qubits
  and set all other qubits to be in the $0$ state.
One can see that $\CCZ^f$ act on the data qubits as $k_{\CCZ}$ copies of CCZ gates.

More conceptually,
  $k_{\CCZ}$ can be defined as the subrank of the 3-tensor obtained by restricting $f$ to the logical operators.
We first describe the 3-tensor, then review the definition of subrank.
Let $\{\z_{1,j'_1}\}_{j'_1 = 1}^{k_1}$ be a set of basis for the X logical operators,
  i.e. the representatives of a basis for the cohomology group $H(Q_1)$.
Define $\{\z_{2,j'_2}\}_{j'_2 = 1}^{k_2}$ and $\{\z_{3,j'_3}\}_{j'_3 = 1}^{k_3}$ similarly.
Let $T$ be the 3-tensor in $\FF_q^{k_1} \otimes \FF_q^{k_2} \otimes \FF_q^{k_3}$,
  where
  \begin{equation}
    T_{j'_1, j'_2, j'_3} = f(\z_{1,j'_1}, \z_{2,j'_2}, \z_{3,j'_3}).
  \end{equation}
One can see that $T$ contains all the information about how $f$ acts on the codespace.

We now review the definition of subrank.
The subrank of a 3-tensor $T \in \FF_q^{k_1} \otimes \FF_q^{k_2} \otimes \FF_q^{k_3}$,
  denoted as $Q(T)$,
  is the largest integer $r$ such that there exist linear maps
  $M_i: \FF_q^{k_i} \to \FF_q^r$ ($r \le k_i$) with the property that
  \begin{equation}
    (M_1 \otimes M_2 \otimes M_3) T = \sum_{j=1}^r e_j \otimes e_j \otimes e_j.
  \end{equation}
It is straightforward to check that the earlier definition of $k_{\CCZ}$ is equal to $Q(T)$.
In particular, the choice of the X logical operators $\z_{1,j}, \z_{2,j}, \z_{3,j}$ for $j \in k_{\CCZ}$
  corresponds to the choice of the linear maps $M_1, M_2, M_3$.

The above discussion leads to the following definition.
\begin{definition}
  Given a CCZ code, we define
  \begin{equation}\label{eq:def-k-CCZ}
    k_{\CCZ} \coloneqq Q(T)
  \end{equation}
  where $T$ is the 3-tensor obtained by restricting $f$ to the logical operators
    and $Q(T)$ is the subrank of $T$.
\end{definition}

\begin{remark}
  In the definition of $k_{\CCZ}$,
    we treat CCZ gates as the desired resource and discard everything else.
  It is possible that the remnant contains nontrivial resource
    which may be recaptured and utilize for better efficiency.
  We leave this possibility for future exploration.
\end{remark}

To summarize, the discussion of CCZ codes generally contains 3 parts:
\begin{enumerate}
  \item Construct three codes $Q_1, Q_2, Q_3$ together with a trilinear map $f$.
  \item Show that the trilinear map $f$ is uniquely defined on the cohomology classes, i.e. $f$ satisfies \Cref{eq:condition-without-trace-main}.
  \item Bound $n_{\CCZ}$ and $k_{\CCZ}$.
\end{enumerate}
In this work, we focus on the first two aspects.
The third aspect will be discussed in the companion paper \cite{golowich2024quantum}.

\section{Interpret logical operators of qLDPC codes as strings and surfaces}
\label{sec:interpret-logical-operators}

It is well-known that in a 2D surface code,
  logical operators can be viewed as 1D strings.
Similarly, in a 3D surface code,
  the Z logical operator corresponds to a 1D string
  and the X logical operator corresponds to a 2D surface.
These geometric interpretations were crucial for understanding
  the intersection number that induces the CCZ gate on a 3D surface code,
  as discussed in \Cref{exam:3D-surface}.
In this section,
  we would like to obtain a similar geometric interpretation for qLDPC codes,
  which will allow us to construct $f$ through the intersection number.

This geometric interpretation for qLDPC codes is largely motivated by
  the recent works on subdivided codes \cite{lin2023geometrically,li2024transform}.
One key step in these works is to impose a 2D structure on qLDPC codes
  and attach 2D surface codes on these 2D regions.
  See \Cref{fig:subdivide}.
This process provides a family of codes with different levels of subdivision characterized by $L$,
  where the original code has $L = 1$.
Since the new code is mostly formed by 2D surface codes,
  the logical operators can be naturally interpreted as 1D strings
    that form a tree, which branches at the place where the surfaces meet.
  See \Cref{fig:string-operator}.

\begin{figure}[H]
  \centering
  \includegraphics[width=0.7\linewidth]{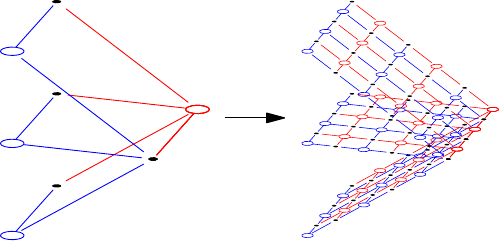}
  \caption{The figure illustrates the subdivision process for part of a quantum code.}
  \label{fig:subdivide}
\end{figure}

\begin{figure}[H]
  \centering
  \includegraphics[width=0.7\linewidth]{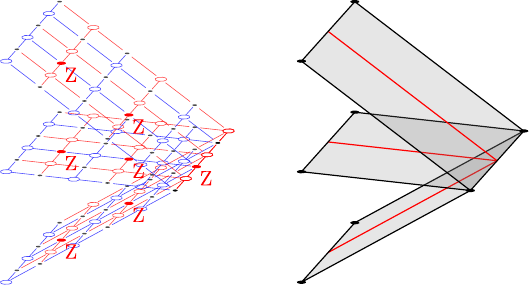}
  \caption{The left figure shows part of a Z logical operator,
            which has a natural string structure
            illustrated on the right.}
  \label{fig:string-operator}
\end{figure}

In a physics language,
  the boundary edge in the figure is condensed by $e_1 e_2 e_3, m_1 m_2, m_2 m_3$
  where $e_i$ and $m_i$ are the electric and magnetic particle on the $i$-th layer.
See \Cref{fig:condensation}.
We say a set of particles (e.g. $e_1 e_2 e_3$) condenses on a boundary
  if this set of particles can be created without extra energy near the boundary.
This is related to the fact that three logical Z strings ending on this boundary
  do not violate any stabilizers.

There are certain criteria on what sets can be condensed.
For general TQFTs, we need to make sure
  every condensing set is `bosonic'
  and different condensing sets mutually commute \cite{kitaev2012models,burnell2018anyon}.
For $\ZZ/2\ZZ$-gauge theories (i.e. toric codes),
  by using the braiding property of the $e$ and $m$ particles,
  the condition can be rephrased as the following.
If we map $e_i$ to $Z_i$ and $m_i$ to $X_i$,
  then as Pauli operators, they commutes.
Indeed $Z_1 Z_2 Z_3, X_1 X_2, X_2 X_3$ mutually commutes.
Because of this mapping,
  we will refer to these boundary condensing patterns as the boundary stabilizers.
If we further require that the boundary is `gapped',
  which is the case for stabilizer codes with large distances,
  every Pauli operator that commutes with the boundary stabilizers
  has to be a boundary stabilizer,
  i.e. the quantum code induced by the boundary stabilizers has $k = 0$.

\begin{figure}[H]
  \centering
  \includegraphics[width=0.4\linewidth]{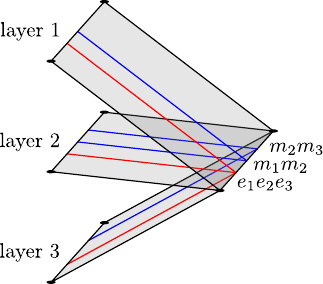}
  \caption{The figure shows three condensing patterns for the code example above.}
  \label{fig:condensation}
\end{figure}

If the code is a CSS code,
  then each condensing set is either all formed by $e$ or $m$.
That means the boundary stabilizers also form a CSS code
  and we can specify the condensing pattern
  using two classical linear codes
  $C_e, C_m \subseteq \FF_q^\Delta$,
  where $\Delta$ is the number of squares meeting at the boundary.
For example, the condensing pattern $e_1 e_2 e_3, m_1 m_2, m_2 m_3$
  has $C_e = \Span\{(1,1,1)\}$,
    $C_m = \Span\{(1,1,0), (0,1,1)\}$.
From the discussion in the last paragraph,
  we observe that when the boundary is `gapped',
  $C_e$ and $C_m$ are the dual to each other,
  $C_m = C_e^\perp$.

In fact, the two classical codes $C_e, C_m$ can be identified with the local codes used in the Tanner code constructions \cite{panteleev2021asymptotically,leverrier2022quantum,dinur2022good}.
In particular, for local codes $C_A, C_B$,
  the horizontal edges have condensing data $C_e = C_A, C_m = C_A^\perp$
  and the vertical edges have condensing data $C_e = C_B, C_m = C_B^\perp$.

As a fun fact, this observation can be used to prove that all three constructions \cite{panteleev2021asymptotically,leverrier2022quantum,dinur2022good}
  are (sparse) chain homotopy equivalent as defined in \Cref{sec:chain-complex}.
This implies that if one of the three codes has certain properties,
  such as large distance or having a decoder,
  then other codes also have such properties.

This observation also elucidate the appearance of dual codes when analyzing Z distance versus X distance of the qLDPC codes.
In the proof of the papers above,
  showing Z distance boils down to certain properties of the local codes $C_A, C_B$,
  while showing X distance boils down to certain properties of the local codes $C_A^\perp, C_B^\perp$.
We can now explain this phenomenon as follows.
The proof of distance is essentially about analyzing the structure of the string patterns.
Since the Z logical operators correspond to strings that ends on the boundaries according to $C_e$,
  proving Z distance only involves $C_e$, i.e. $C_A, C_B$.
Similarly, since the X logical operators correspond to strings that ends on the boundaries according $C_m$,
  proving X distance only involves $C_m$, i.e. $C_A^\perp, C_B^\perp$.

More generally, the discussion about boundary conditions
  applies to codes based on 2D simplicial complexes
    such as the quantum variant of \cite{dinur2023new},
  to codes based on higher dimensional cubical complexes \cite{dinur2024expansion},
  and to sheaf codes \cite{first2022good,panteleev2024maximally}.
The formal discussion for sheaf codes will appear in \Cref{sec:sheaf-code}.

In this work,
  the geometric interpretation of the logical operators
  is used to define $f$
  as the intersection number.
To achieve this,
  we will formalize this relation
  by providing an explicit mapping for quantum codes constructed from 2D and 3D cubical complexes.

\subsection{An explicit map from Pauli X operators to geometric objects for 2D and 3D cubical codes}
\label{sec:map-logical-to-geometry}

For the purpose of constructing CZ and CCZ gates,
  all we need is a mapping for the Pauli X operators.
This makes the task especially simple if we place the checks and qubits geometrically,
  which is what we did in \Cref{sec:cubical-code}.

We start with quantum codes based on square complexes
  $\cC(G, \{A_1, A_2\}, \{C_1, C_2\})$.
Given a vector on the qudits
  $\alpha \in \cC^1(G, \{A_1, A_2\}, \{C_1, C_2\})$
  that represents a Pauli X operator.
The local coefficient on each edge $e$,
  $\alpha(e) \in \cF_e$,
  corresponds to a function on the $\Delta$ adjacent faces
  $h_i^T \alpha(e) \in \cG_e \subseteq \FF_q^\Delta$.
These values can be attached to the segments
  between the midpoint of the edge $e$ and the center of the corresponding face $f$.
(This geometric object should be thought of as the chain in algebraic topology.)
See \Cref{fig:2D-geometric}.

\begin{figure}[H]
  \centering
  \includegraphics[width=0.45\linewidth]{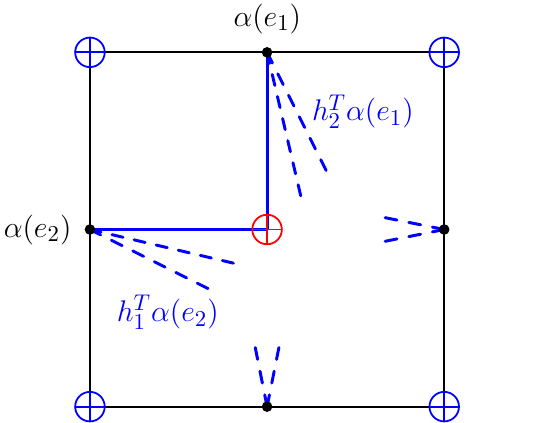}
  \caption{Each edge has a local coefficient $\alpha(e) \in V_e$.
            Such a local coefficient induces a vector $h_i^T \alpha(e) \in \FF_q^\Delta$.
            The vector $h_i^T \alpha(e)$ is then used to induce segments between
            the midpoints of the edges and the centers of the faces.
            Notice that because $\delta \alpha = 0$,
              the sum of the segment values on a face is $0$.}
  \label{fig:2D-geometric}
\end{figure}

We observe three properties of this mapping.
First, the string pattern near the boundary belongs to the local codes $C_1$ or $C_2$.
Second, the condition of a X logical operator $\delta \alpha = 0$,
  corresponds to the condition that the sum of the segment values on each face is $0$,
  because $\delta \alpha$ is obtained by applying the $\cores$ map on $\alpha(e)$,
    which is exactly applying $h_i^T$ on $\alpha(e)$.
Finally, the X stabilizer generators correspond to logical operators
  that are supported locally on the segments around a vertex.

A similar mapping exists for 3D cubical codes $\cC(G, \{A_1, A_2, A_3\}, \{C_1, C_2, C_3\}$).
Given a vector on the qudits
  $\alpha \in \cC^1(G, \{A_1, A_2, A_3\}, \{C_1, C_2, C_3\})$.
The local coefficient on each edge $e$,
  $\alpha(e) \in \cF_e$,
  corresponds to a function on the $\Delta^2$ adjacent cubes
  $(h_i^T \otimes h_j^T) \alpha(e) \in \cG_e \subseteq \FF_q^{\Delta^2}$.
These values can be attached to the faces
  spanned by the midpoint of the edge $e$,
  the two center of the corresponding neighboring faces,
  and the center of the corresponding cube.
See \Cref{fig:3D-geometric}.

\begin{figure}
  \centering
  \includegraphics[width=0.7\linewidth]{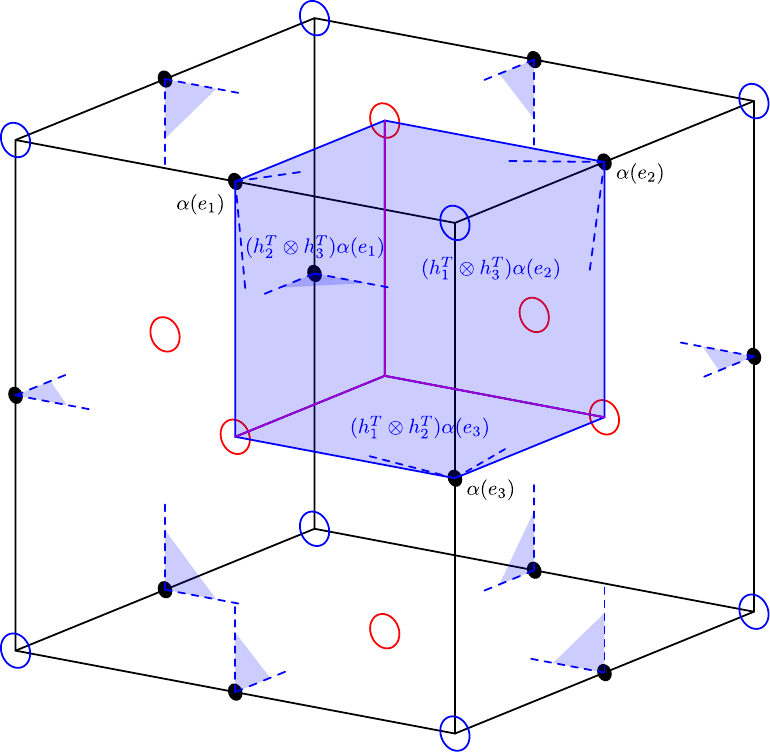}
  \caption{Each edge has a local coefficient $\alpha(e) \in V_e$.
            Such a local coefficient induces a vector $(h_i^T \otimes h_j^T) \alpha(e) \in \FF_q^{\Delta^2}$.
            The vector $(h_i^T \otimes h_j^T) \alpha(e)$ is then used to
              induce the blue surfaces.
            Notice that because $\delta \alpha = 0$,
              the sum of the surface values near a purple line is $0$.}
  \label{fig:3D-geometric}
\end{figure}

Similarly, we observe three properties of this mapping.
First, the surface pattern near the boundary belongs to the local codes $C_i$.
See \Cref{fig:3D-boundary-code}.
The $\Delta^2$ values induced from an edge belongs to $C_i \otimes C_j$.
When viewed from the 2D boundaries perpendicular to direction $i$,
  this splits into $\Delta$ vectors in $C_i$ on $\Delta$ boundaries.
Second, the condition of a X logical operator $\delta \alpha = 0$,
  corresponds to the condition that the sum of the surface values near
    an edge connecting the center of a face (Z check) and the center of the cube is $0$.
  See \Cref{fig:3D-geometric}.
Finally, the X stabilizer generators correspond to logical operators
  that are supported locally on the surfaces around a vertex.
This description leads to an explicit way to view X logical operators as surfaces,
  which will be used to define the trilinear function $f$ in the next section.

\begin{figure}
  \centering
  \includegraphics[width=0.5\linewidth]{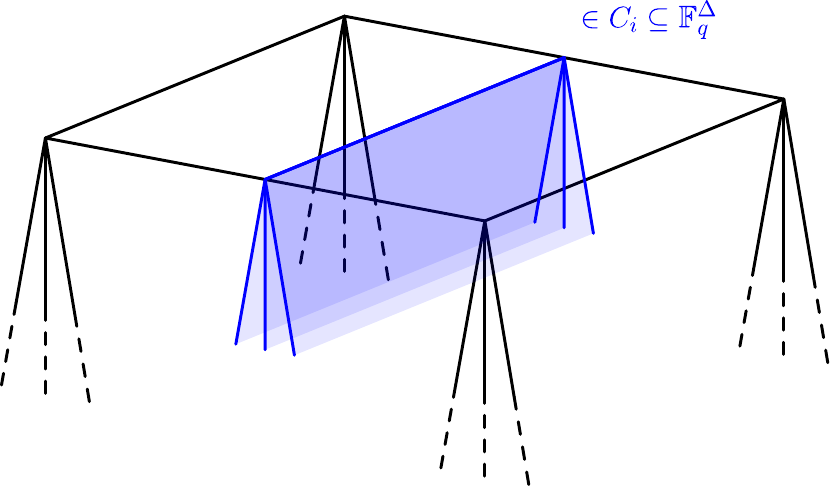}
  \caption{There are $\Delta$ cubes incident to a top boundary face.
            Each such cube carries a value.
            These values form a codeword in $C_i$.}
  \label{fig:3D-boundary-code}
\end{figure}

\section{Transversal CCZ operations on 3D cubical codes}
\label{sec:CCZ-cubical-code}

The goal of this section is to construct a trilinear function $f$ for the 3D cubical code,
  which as discussed in \Cref{sec:CCZ-formalism},
  leads to transversal CCZ operations.
This construction is based on the geometrical interpretation of logical operators described in the previous section.
We first present a simpler case for the 2D square code,
  where we construct a bilinear function
  that is defined on the cohomology classes.
We then generalize this idea to the 3D cubical code.
Since the discussion involves three codes,
  we will label the three codes using $1, 2, 3$,
  and label the directions using $u, v, w$
  in this section.

\subsection{Warm up: a bilinear map on a 2D square code}

Let $X$ be the square complex $X(G, \{A_u, A_v\})$.
Consider two quantum codes based on the square complex
\begin{equation}
  Q_1 = Q(X, \{C^1_u, C^1_v\}, \ell=1),\quad Q_2 = Q(X, \{C^2_u, C^2_v\}, \ell=1).
\end{equation}
The goal is to construct a bilinear map
  $f: \cC^1(X, \{C^1_u, C^1_v\}) \times \cC^1(X, \{C^2_u, C^2_v\}) \to \FF_q$
such that
\begin{equation}\label{eq:condition-bilinear-version}
  f(\z_1, \z_2) = f(\z_1 + \b_1, \z_2 + \b_2)
\end{equation}
for all $\z_i \in Z^1(X, \{C^i_u, C^i_v\})$
  and $\b_i \in B^1(X, \{C^i_u, C^i_v\})$.

To achieve this, we require that
  $C_u^1 \odot C_u^2$ and $C_v^1 \odot C_v^2$ are both contained
  within the parity check code,
  where $C^1 \odot C^2 = \Span\{c_1 \odot c_2: c_1 \in C^1, c_2 \in C^2\}$
  is the vector space generated by $c_1 \odot c_2$
  and $c_1 \odot c_2$ is the entrywise product of $c_1, c_2$.
That means, for every $c_1 \in C^1_u, c_2 \in C^2_u$,
  we have $\sum_{i=1}^\Delta (c_1)_i (c_2)_i = 0$
  where $(c_1)_i$ is the $i$-th entry of $c_1$.

We now construct the bilinear map $f: \cC^1(X, \{C^1_u, C^1_v\}) \times \cC^1(X, \{C^2_u, C^2_v\}) \to \FF_q$.
  The goal is to map the vectors
    $\alpha_1 \in \cC^1(X, \{C^1_u, C^1_v\})$
    and $\alpha_2 \in \cC^1(X, \{C^2_u, C^2_v\})$
    into geometric strings, as described in
    \Cref{sec:map-logical-to-geometry},
    and then define $f$ as the intersection number of these strings.
The appeared difficulty is that the strings overlap,
  making it unclear how to define their intersection number directly.
The solution is well-understood in mathematics,
  which is to perturb the strings into a ``generic position''.
In our case, we can achieve this by shifting the string as shown in \Cref{fig:2D-intersection}.
Note that these strings have associated values attached.
When computing the intersection,
  we multiply the attached values.
By summing up the contributions from all faces,
  we have the definition
\begin{align}
  f(\alpha_1, \alpha_2)
  \coloneqq \sum_{g \in G, a_u \in A_u, a_v \in A_v}
  \cores_{e_{*1},sq}(\alpha_1(e_{*1})) \; \cores_{e_{0*},sq}(\alpha_2(e_{0*})) \notag \\
  + \cores_{e_{1*},sq}(\alpha_1(e_{1*})) \; \cores_{e_{*0},sq}(\alpha_2(e_{*0}))
\end{align}
where $e_{*0} = (g; a_u, 0)$, $e_{0*} = (g; 0, a_v)$,
  $e_{*1} = (a_v g; a_u, 1)$, $e_{1*} = (a_u g; 1, a_v)$,
  $sq=(g; a_u, a_v)$.\footnote{
    For general finite fields $\FF_q$ with characteristic other than $2$,
    the definition has to be modified to incorporate certain negative signs.
  }

\begin{figure}[H]
  \centering
  \includegraphics[width=0.6\linewidth]{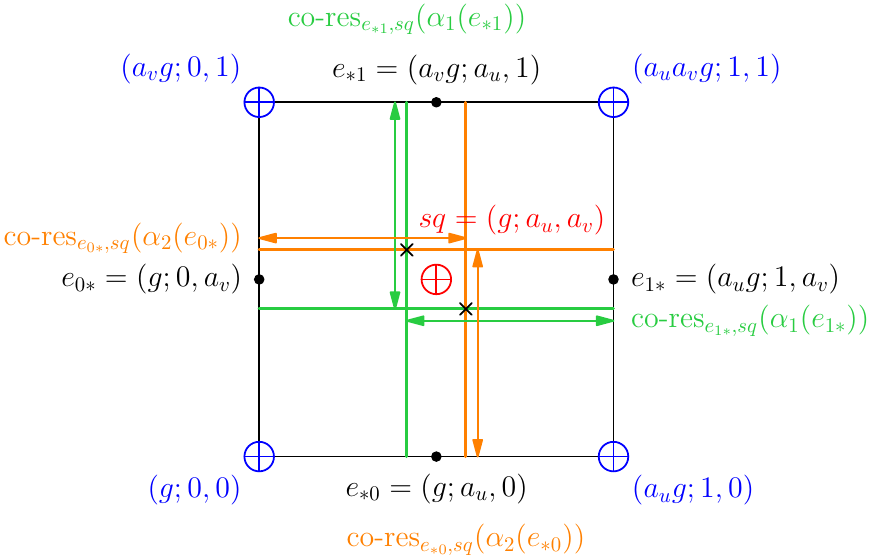}
  \caption{The figure shows two quantum codes constructed based on a square complex,
            $Q_1 = Q(X, \{C^1_u, C^1_v\}, \ell=1)$ and $Q_2 = Q(X, \{C^2_u, C^2_v\}, \ell=1)$.
            The associated Pauli X operators are colored in green and in orange.
            On each face, there are two contributions to the intersection number marked with a cross.
            These corresponds to the two terms
            $\cores_{e_{*1},sq}(\alpha_1(e_{*1})) \cores_{e_{0*},sq}(\alpha_2(e_{0*}))$
            and $\cores_{e_{1*},sq}(\alpha_1(e_{1*})) \cores_{e_{*0},sq}(\alpha_2(e_{*0}))$.}
  \label{fig:2D-intersection}
\end{figure}

What remains to be done is to verify that $f$ is well-defined on the cohomology classes,
  meaning that it satisfies \Cref{eq:condition-bilinear-version}.
This is where
  the assumption that both $C_a^1 \odot C_a^2$ and $C_b^1 \odot C_b^2$
  are contained within the parity check code comes into play.
The following discussion aims to provide intuition for this result,
  while the formal proof is deferred to \Cref{sec:sheaf-code}.

\begin{figure}
  \centering
  \includegraphics[width=0.6\linewidth]{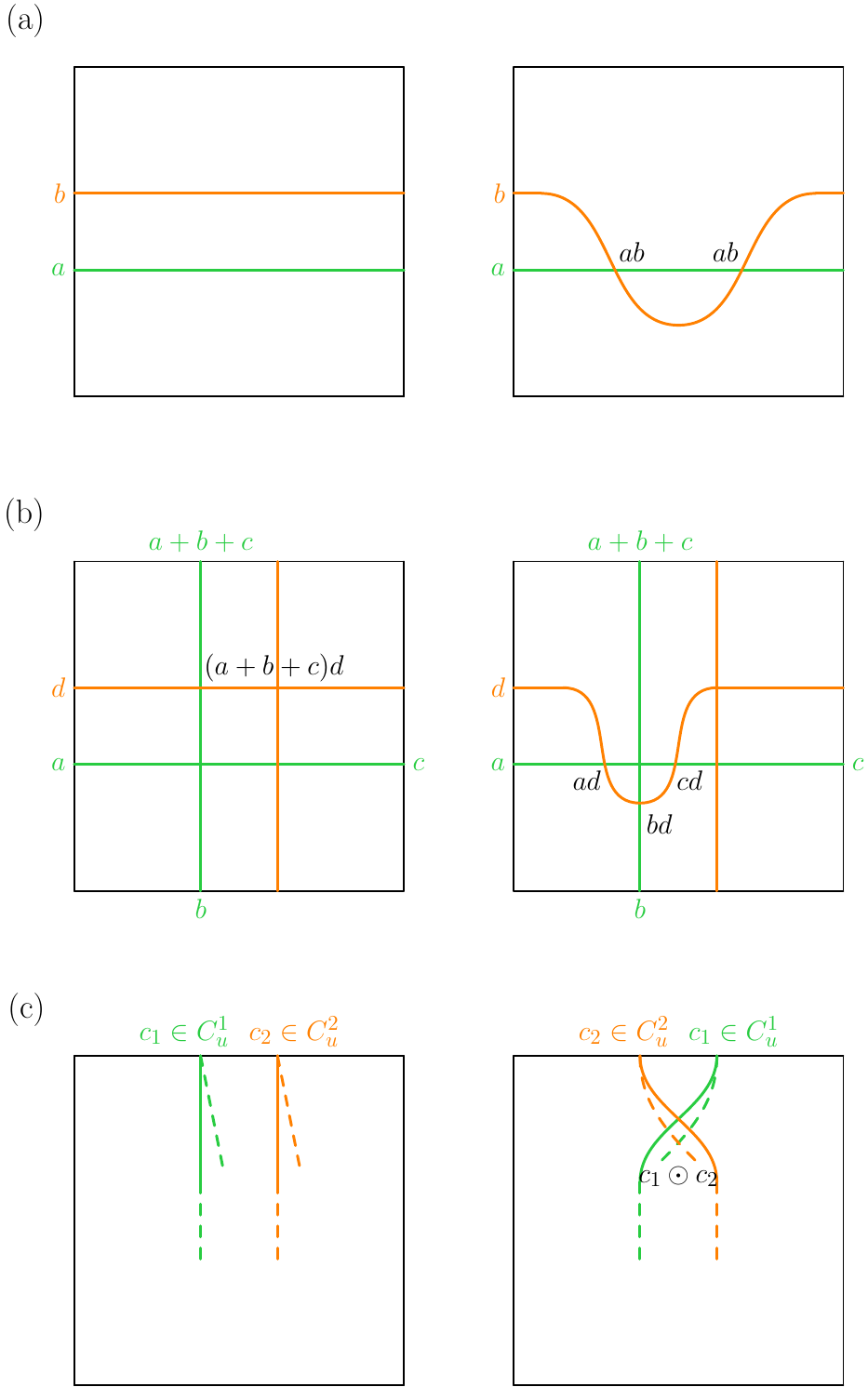}
  \vspace{2em}
  \caption{The three main types of topological deformations.
            We observe that the intersection number is invariant under those deformations.
            This explains why $f$ is well-defined on the cohomology classes.}
  \label{fig:2D-invariant}
\end{figure}

In the language of topology,
  the condition in \Cref{eq:condition-bilinear-version}
  is saying that the intersection number of two logical operators $\z_1, \z_2$ is invariant under topological deformations.
There are mainly three types of deformations, as illustrated in \Cref{fig:2D-invariant}.
The invariance of the first type holds
  for the simple reason that the intersections have the same value.
The invariance of the second type follows from
  the property that, for logical operators,
  the four string values on a face sums to $0$.
The invariance of the third type is guaranteed by
  the fact that the string patterns near the boundary
  are elements of the local codes,
  $c_1 \in C_u^1, c_2 \in C_u^2$.
The change in the intersection number
  is given by $\sum_{i=1}^\Delta (c_1)_i (c_2)_i$,
  which equals to $0$
  because $C_u^1 \odot C_u^2$ is contained within the parity check code.
We also require $C_v^1 \odot C_v^2$ to be contained within the parity check code
  to ensure invariance for similar topological deformations along the vertical boundaries.

\subsection{A trilinear map on a 3D cubical code}

The same idea of using intersections to define $f$ can be extended to 3D cubical codes.
Let $X$ be the cubical complex $X(G, \{A_u, A_v, A_w\})$.
Consider three quantum codes based on the cubical complex
\begin{equation}
  Q_1 = Q(X, \{C^1_u, C^1_v, C^1_w\}, \ell=1),\quad Q_2 = Q(X, \{C^2_u, C^2_v, C^2_w\}, \ell=1), \quad Q_3 = Q(X, \{C^3_u, C^3_v, C^3_w\}, \ell=1).
\end{equation}
The goal is to construct a trilinear map
  $f: \cC^1(X, \{C^1_u, C^1_v, C^1_w\}) \times \cC^1(X, \{C^2_u, C^2_v, C^2_w\}) \times \cC^1(X, \{C^3_u, C^3_v, C^3_w\}) \to \FF_q$
such that
\begin{equation}
  f(\z_1, \z_2, \z_3) = f(\z_1 + \b_1, \z_2 + \b_2, \z_3 + \b_3)
\end{equation}
for all $\z_i \in Z^1(X, \{C^i_u, C^i_v, C^i_w\})$
  and $\b_i \in B^1(X, \{C^i_u, C^i_v, C^i_w\})$.
As discussed in \Cref{sec:CCZ-formalism},
  these data lead to a transversal CCZ operation on the codes.

To achieve this, we require that
  $C_u^1 \odot C_u^2 \odot C_u^3$,
  $C_v^1 \odot C_v^2 \odot C_v^3$,
  $C_w^1 \odot C_w^2 \odot C_w^3$,
  are all contained
  within the parity check code.
That means, for every $c_1 \in C^1_u, c_2 \in C^2_u, c_3 \in C^3_u$,
  we have $\sum_{i=1}^\Delta (c_1)_i (c_2)_i (c_3)_i = 0$
  where $(c_1)_i$ is the $i$-th entry of $c_1$.
This can be satisfied, for example,
  by taking low-degree Reed-Solomon codes over $\FF_q$ with $\Delta = q$.

We can similarly construct the trilinear map $f$
  by mapping the vectors
  $\alpha_i \in \cC^i(X, \{C^i_u, C^i_v, C^i_w\})$
  into geometric surfaces, as described in
  \Cref{sec:map-logical-to-geometry},
  and then define $f$ as the triple intersection number of these surfaces.
We again shift the surfaces slightly
  to obtain the triple intersection number
  as illustrated in \Cref{fig:3D-intersection}.
Overall, we have
\begin{align*}
  f(\alpha_1, \alpha_2, \alpha_3)
  = &\sum_{g \in G, a_u \in A_u, a_v \in A_v, a_w \in A_w} \\
  \;& \cores_{e_{...},cu}(\a_1(a_va_wg; a_u, 1, 1)) \; \cores_{e_{...},cu}(\a_2(a_wg; 0, a_v, 1)) \; \cores_{e_{...},cu}(\a_3(g; 0, 0, a_w)) \\
  +\;& \cores_{e_{...},cu}(\a_1(a_va_wg; a_u, 1, 1)) \; \cores_{e_{...},cu}(\a_2(a_vg; 0, 1, a_w)) \; \cores_{e_{...},cu}(\a_3(g; 0, a_v, 0)) \\
  +\;& \cores_{e_{...},cu}(\a_1(a_ua_wg; 1, a_v, 1)) \; \cores_{e_{...},cu}(\a_2(a_ug; 1, 0, a_w)) \; \cores_{e_{...},cu}(\a_3(g; a_u, 0, 0)) \\
  +\;& \cores_{e_{...},cu}(\a_1(a_ua_wg; 1, a_v, 1)) \; \cores_{e_{...},cu}(\a_2(a_wg; a_u, 0, 1)) \; \cores_{e_{...},cu}(\a_3(g; 0, 0, a_w)) \\
  +\;& \cores_{e_{...},cu}(\a_1(a_ua_vg; 1, 1, a_w)) \; \cores_{e_{...},cu}(\a_2(a_vg; a_u, 1, 0)) \; \cores_{e_{...},cu}(\a_3(g; 0, a_v, 0)) \\
  +\;& \cores_{e_{...},cu}(\a_1(a_ua_vg; 1, 1, a_w)) \; \cores_{e_{...},cu}(\a_2(a_ug; 1, a_v, 0)) \; \cores_{e_{...},cu}(\a_3(g; a_u, 0, 0))
\end{align*}
where $e_{...}$ is the edge in the argument of $\alpha_i$,
  and $cu = (g; a_u, a_v, a_w)$.
For example, in the first term, $e_{...} = (a_va_wg; a_u, 1, 1)$.
Although these expressions are explicit,
  they are cumbersome to work with.
Therefore, we will reformulate the definition of $f$ more succinctly using the language of sheaves in \Cref{sec:sheaf-code}.

\begin{figure}
  \centering
  \includegraphics[width=0.7\linewidth]{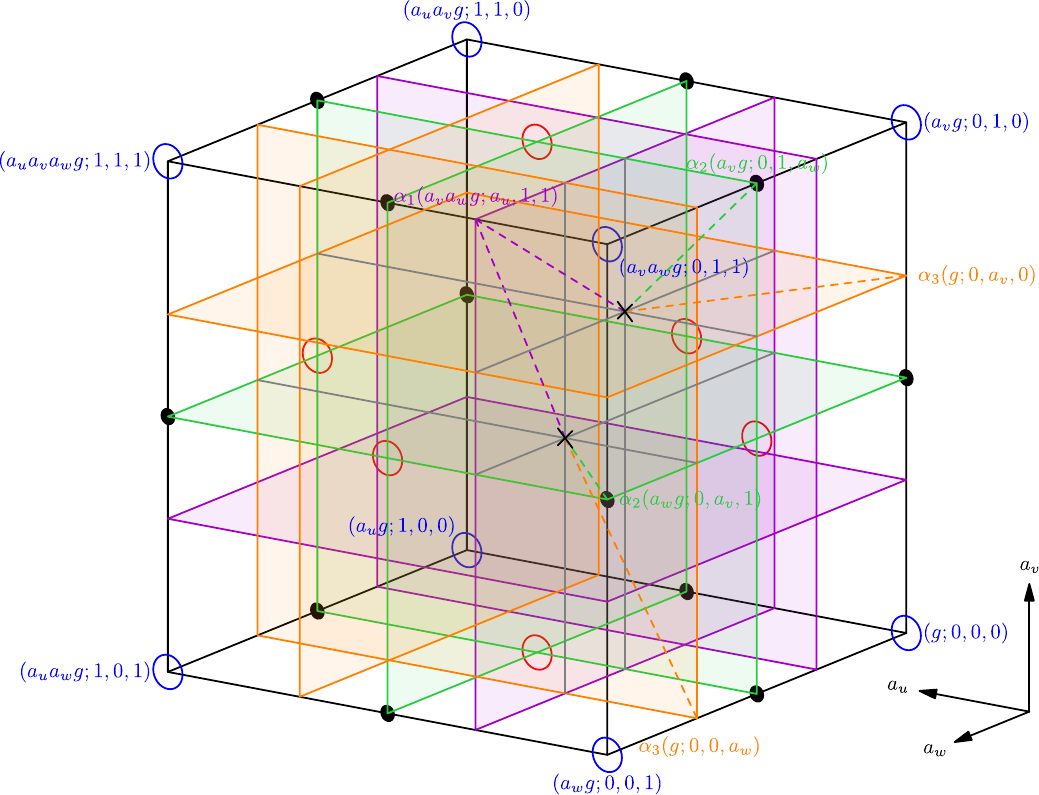}
  \caption{The figure shows three quantum codes constructed based on a cubical complex,
            $Q_1, Q_2, Q_3$.
            The associated Pauli X operators are colored in purple, green, and orange respectively.
            In each cube, there are a total of $6$ intersections
              which contribute to the $6$ terms in the expression of $f$.
            We draw the first two contributions in the figure
              marked by the crosses.}
  \label{fig:3D-intersection}
\end{figure}

Analogous to the discussion in the previous section,
  for $f$ to be well-defined on the cohomology classes,
  we want the intersection number to remain invariant under of the topological deformations.
One of the nontrivial deformations occurs near the boundary,
  where we deform the surface as shown in \Cref{fig:3D-invariant}.
With some imagination,
  we observe that each cube introduces an extra triple intersection
  with the value $(c_1)_i (c_2)_i (c_3)_i$,
  where $c_1 \in C^1_u, c_2 \in C^2_u, c_3 \in C^3_u$
  correspond to the patterns near the boundary.
This means that the intersection number is changed by
  $\sum_{i=1}^\Delta (c_1)_i (c_2)_i (c_3)_i$, which evaluates to $0$
  due to the assumption that $C^1_u \odot C^2_u \odot C^3_u$ is contained within the parity check code.

\begin{figure}
  \centering
  \includegraphics[width=0.6\linewidth]{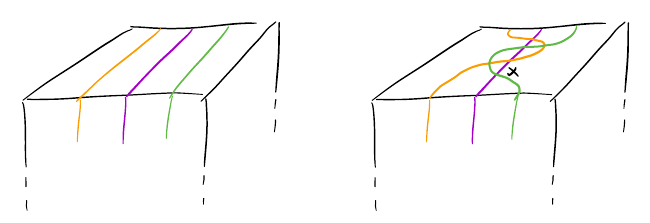}
  \caption{A topological deformation where the invariance of $f$
            relies on the assumption that $C^1_u \odot C^2_u \odot C^3_u$ is contained within the parity check code.}
  \label{fig:3D-invariant}
\end{figure}

\section{Transversal CCZ operations on sheaf codes} \label{sec:sheaf-code}

In this section, we prove the technical result that enables
  transversal CCZ operations on sheaf codes.
This includes the cubical code
  discussed in the previous section.
We begin by reviewing cellular complexes and sheaf complexes,
This is followed by a discussion on how to map Pauli X operators geometrically.
Finally, we formalize these intuitions by defining the cup product
  on all sheaf codes.

We remark that there have are several variants of sheaves that have been considered
  \cite{curry2014sheaves,first2022good,panteleev2024maximally}.
There is the vanilla sheaf code in \cite{curry2014sheaves,first2022good}
  and the Tanner sheaf code and maximally extendable sheaf code defined in \cite{panteleev2024maximally}.
In this work, we suggest a narrower definition of sheaf codes
  compared to the definitions in \cite{curry2014sheaves}.
In fact, our definition almost agrees with the definition of Tanner sheaf code in \cite{panteleev2024maximally},
  except a technical difference that will be explained later.

\subsection{Cellular complexes}

Let $X$ be a $t$-dimensional cellular complex,
  where $X(i)$ denotes the set of $i$-cells.
Given two cells $\sigma, \tau$,
  we write $\sigma \preceq \tau$ if $\sigma$ is a subcell of $\tau$,
  and write $\sigma \precdot \tau$ if $\sigma$ is an `immediate' subcell of $\tau$
    where their dimensions differ by $1$.

Mathematically speaking,
  it is easier to formalize the notion of a simplicial complexes,
  then obtain a cellular complex by merging the simplices.
However, we will skip those details.

In the later discussion we assume that
  for every pair of $i$-cell $\sigma$ and $(i+2)$-cell $\tau$,
  such that $\sigma \prec \tau$,
  there are exactly two $(i+1)$-cells $\pi$ such that $\sigma \precdot \pi \precdot \tau$.
This simplifies the discussion on the gluability of a sheaf.
This property generally holds with enough subdivision of a cellular complex
  and is true for the cubical complexes discussed earlier.
In fact, this property can be shown rigorously for simplicial complexes.
(Note that it is generally assumed for simplicial complexes that
  the vertices of each simplex are distinct.)
\begin{claim} \label{claim:two}
  Given a simplicial complex.
  For every pair of $i$-simplex $\sigma$ and $(i+2)$-simplex $\tau$,
  such that $\sigma \prec \tau$,
  there are exactly two $(i+1)$-simplices $\pi$ such that $\sigma \precdot \pi \precdot \tau$.
\end{claim}
\begin{proof}
  Let $V$ be the vertices of $\sigma$ and $W$ be the vertices of $\tau$.
  Since the vertices of each cell are distinct,
    $W \backslash V$ has size $2$.
  Thus, the two $(i+1)$-cells $\pi$ are specified by the vertices $V \cup \{v\}$ for $v \in W \backslash V$.
\end{proof}

We introduce some notations.
$X_{\ge \sigma}$ denotes the sub-complex consists of all cells above $\sigma$
\begin{equation}
  X_{\ge \sigma} = \{\tau \in X: \tau \succeq \sigma\}.
\end{equation}
$X_{\le \sigma}$ denotes the sub-complex consists of all cells below $\sigma$
\begin{equation}
  X_{\le \sigma} = \{\tau \in X: \tau \preceq \sigma\}.
\end{equation}
We define $X_{\ge \sigma}(k) = X_{\ge \sigma} \cap X(k)$
  and $X_{\le \sigma}(k) = X_{\le \sigma} \cap X(k)$.

For the purpose of constructing sheaves in the next section,
  it will be useful to keep the following analogies in mind.
We think of each $t$-cell as a point,
  and each $i$-cell $\sigma$ as an ``open'' set that contains the points $X_{\ge \sigma}(t)$.
Therefore, the relation of set inclusion $X_{\ge \sigma}(t) \supseteq X_{\ge \tau}(t)$,
  corresponds to the relation $\sigma \preceq \tau$.
Notice that these ``open'' sets do not form a standard topology,
  as they do not satisfy the property that the union of every family of open sets is open.
Instead, they form the Alexandrov topology,
  which only requires that the intersection of every family of open sets is open.
This perspective of Alexandrov topology is discussed in \cite{curry2014sheaves} and \cite{panteleev2024maximally},
  although we do not use this perspective in the later discussions.

\subsection{Sheaves on cellular complexes}

Sheaf theory is a well-established topic in mathematics.
Recently, chain complexes induced from sheaves have emerged as a natural framework for constructing codes
  \cite{first2022good,dinur2023new,dinur2024expansion,panteleev2024maximally}.
For readers unfamiliar with sheaf theory,
  it can be helpful to consult chapter 2 in \cite{vakil2017rising}
  to understand the motivation behind sheaves,
  which are traditionally defined over topological spaces.
It is known that sheaves over topological spaces have several categorical properties
  and is the right object to study in various settings.

It is important to note, however, that the definition of sheaves over cellular complexes
  differs somewhat from the traditional definition of sheaves over topological spaces.
Therefore, one should be careful when drawing parallels between
  established results on sheaves over topological spaces and similar results for cellular complexes.
In fact, the definition of sheaves on cellular complex introduced in \cite{curry2014sheaves},
  seems more similar to the traditional definition of presheaves in some respects\footnote{
    The reason the author refers to it as a sheaf is because the definition satisfies certain categorical properties.
    As of my current understanding, the author may have implicitly applied certain sheafification to make the object more well-behaved.
  },
  and one might wonder whether it is the right object to study in the context of error correcting codes.

These concerns motivate us to introduce additional conditions on sheaves over cellular complexes.
The conditions we impose have certain analogs in the traditional definition of sheaves
  and could be useful for certain applications.


We begin by reviewing the definition of sheaves over cellular complexes as introduced in \cite{curry2014sheaves}.
To distinguish this from our definition of sheaves,
  we refer to this structure as a ``presheaf'' with quotation marks.
A ``presheaf'' $\cF$ on a cellular complex $X$ is the following data.
\begin{itemize}
  \item For each $i$-cell $\sigma \in X(i)$,
    we have a vector space over $\FF_q$, $\cF_\sigma$.
    Recall that each cell corresponds to a set $X_{\ge \sigma}(t)$,
      and we call the elements of $\cF_\sigma$ as sections of $\cF$ over $X_{\ge \sigma}(t)$.
  \item For each inclusion $\sigma \preceq \tau$ (i.e. $X_{\ge \sigma}(t) \supseteq X_{\ge \tau}(t)$),
    we have a restriction map $\res_{\sigma, \tau}: \cF_\sigma \to \cF_\tau$.
\end{itemize}
The data is required to satisfy the following two conditions.
\begin{itemize}
  \item The map $\res_{\sigma, \sigma}$ is the identity: $\res_{\sigma, \sigma} = \id_{\cF_\sigma}$.
  \item If $\sigma \preceq \pi \preceq \tau$ are inclusions,
    then the restriction maps commute,
    i.e. $\res_{\sigma, \tau} = \res_{\pi, \tau} \res_{\sigma, \pi}$.
\end{itemize}
When the context is clear,
  we write $|_\tau$ in place of $\res_{\sigma, \tau}$.

A useful way to think about an element in $\cF_\sigma$
  is to view it as a function defined on $X_{\ge \sigma}(t)$.
In particular, given $c \in \cF_\sigma$,
  we can define the value of $c$ at $\tau \in X_{\ge \sigma}(t)$
    as $c(\tau) = \res_{\sigma, \tau}(c)$.
This explains why we call the elements of $\cF_\sigma$ as sections of $\cF$ over $X_{\ge \sigma}(t)$.
However, there is a caveat to this interpretation.
In particular, the map from $\cF_\sigma$ to functions on $X_{\ge \sigma}(t)$ is not necessarily injective for ``presheaves''.
That means that some information in $\cF_\sigma$ may be lost
  when viewed as a function on $X_{\ge \sigma}(t)$.
This is one reason why additional conditions are required
  to guarantee desireable behaviors for ``presheaves''.
As we will see, the conditions we impose
  imply that the map is injective,
  so that the element in $\cF_\sigma$ one-one corresponds to a function on $X_{\ge \sigma}(t)$.

Even though we refer to the object as a ``presheaf'',
  it differs from traditional presheaves over topological spaces
  in several ways.
While there appears to be a dictionary between
  open sets in topological spaces and
  the regions $X_{\ge \sigma}(t)$ for some cell $\sigma$,
  this analogy has some subtleties.
For example, in topology, the union of open sets is an open set,
  but the union of regions $X_{\ge \sigma}(t)$ is generally not of the form $X_{\ge \tau}(t)$.
The presents a problem: a presheaf should assign a vector space to each open set,
  yet the current definition of a ``presheaf'' over a cellular complex
  does not assign a vector space to the union of such regions.
This indicates that the definition of ``presheaves'' over cellular complexes
  lacks some nonlocal data
  compared to presheaves over topological spaces.
This is why we put the quotation marks around ``presheaves''.

The definition of sheaves is designed to address these two issues.
The first issue is resolved by the \emph{identity axiom},
  which ensures that sections are uniquely determined by their local values.
The second issue is resolved by the \emph{gluability axiom},
  which states that local sections can be glued together to form sections over larger sets.
We first review these two conditions for sheaves over topological spaces.
\begin{itemize}
  \item Identity axiom: If $\{U_i\}_{i \in I}$ is an open cover of $U$,
        and $c_1, c_2 \in \cF_U$,
        and $\res_{U, U_i} c_1 = \res_{U, U_i} c_2$ for all $i$,
        then $c_1 = c_2$.
  \item Gluability axiom: If $\{U_i\}_{i \in I}$ is an open cover of $U$, then given $c_i \in \cF_{U_i}$ for all $i$,
        such that $\res_{U_i, U_i \cap U_j} c_i = \res_{U_j, U_i \cap U_j} c_j$ for all $i, j$,
        then there is some $c \in \cF_U$
        such that $\res_{U, U_i} c = c_i$ for all $i$.
\end{itemize}
These two conditions can be interpreted as the exactness of the following chain complex
\begin{equation}\label{eq:topological-sheaf}
  0 \to \cF_U \to \prod_{i \in I} \cF_{U_i} \to \prod_{i, j \in I, i \ne j} \cF_{U_i \cap U_j}.
\end{equation}
The second map is $c \mapsto \{\res_{U, U_i} (c)\}_{i \in I}$
  and the last map is $\{c_i\}_{i \in I} \mapsto \{c_i - c_j\}_{i, j \in I, i \ne j}$,
  where we impose an ordering on the elements in $I$.
Identity is exactness at $\cF_U$
  and gluability is exactness at $\prod_i \cF_{U_i}$.

In the context of topological spaces,
  these conditions must be satisfied by all possible open sets.
However, for cellular complexes,
  the analogous notion of open sets is a bit different.
So instead we impose the following condition:
for every $i$-cell $\sigma \in X$ with $i \le t-2$,
  the following chain complex is exact
\begin{equation}\label{eq:cellular-complex-sheaf}
  0
  \to \cF_\sigma
  \to \prod_{\sigma_{i+1} \in X_{\ge \sigma}(i+1)} \cF_{\sigma_{i+1}}
  \to \prod_{\sigma_{i+2} \in X_{\ge \sigma}(i+2)} \cF_{\sigma_{i+2}}
\end{equation}
and for $i$-cell $\sigma \in X$ with $i = t-1$,
  the following chain complex is exact
\begin{equation}
  0
  \to \cF_\sigma
  \to \prod_{\sigma_{i+1} \in X_{\ge \sigma}(i+1)} \cF_{\sigma_{i+1}}.
\end{equation}
We will refer to the exactness at $\cF_\sigma$ as the identity axiom on $\sigma$
  and the exactness at $\prod_{\sigma_{i+1} \in X_{\ge \sigma}(i+1)} \cF_{\sigma_{i+1}}$ as the gluability axiom on $\sigma$.

The conditions we impose follow closely with the definition of sheaves over topological spaces,
  with the difference being that we now restrict ourselves to a specific family of ``open covers''
  and a particular set of intersections between the ``open sets''.
It is worth noting that in the traditional setting
  each intersection $U_i \cap U_j$ corresponds to exactly two open sets $U_i$ and $U_j$,
  which ensures that the sections on $U_i$ and $U_j$ are consistent.
Similarly, in our context, for every $(i+2)$-cell $\sigma_{i+2} \in X_{\ge \sigma}(i+2)$,
  there are exactly two $(i+1)$-cells $\sigma_{i+1}$ such that $\sigma \precdot \sigma_{i+1} \precdot \sigma_{i+2}$ (discussed near \Cref{claim:two}),
  which guarantees that the section on the two $(i+1)$-cells are consistent.
In fact, when $X$ is a simplicial complex,
  and we take $\{\sigma_{i+1}\}_{\sigma_{i+1} \in X_{\ge \sigma}(i+1)}$ as an open cover of $\sigma$,
  the two expressions in \Cref{eq:cellular-complex-sheaf} and \Cref{eq:topological-sheaf}
  are exactly the same.

We now describe an explicit characterization of $\cF$,
  as a consequence of the sheaf conditions.
\begin{lemma}\label{lem:structure-theorem-1}
  For every cell $\sigma \in X$,
  the map $\cF_{\sigma} \xrightarrow{\iota_\sigma} \prod_{\tau \in X_{\ge \sigma}(t)} \cF_{\tau}$,
  which sends $c \mapsto \prod_{\tau \in X_{\ge \sigma}(t)} \res_{\sigma, \tau}(c)$,
  is injective.

  Furthermore, the restriction map $\res_{\sigma, \pi}$ corresponds to the restriction map on $X(t)$
  \begin{equation}
    \hspace{2em}
    \begin{tikzcd}
      \cF_\sigma \arrow{r}{\iota_\sigma} \arrow[swap]{d}{\res_{\sigma, \pi}} & \prod_{\tau \in X_{\ge \sigma}(t)} \cF_{\tau} \arrow{d}{\res_{X_{\ge \sigma}(t), X_{\ge \pi}(t)}} \\
      \cF_\pi \arrow{r}{\iota_\pi} & \prod_{\tau \in X_{\ge \pi}(t)} \cF_{\tau}
    \end{tikzcd}
  \end{equation}
  where $\res_{X_{\ge \sigma}(t), X_{\ge \pi}(t)}$ is the map that restricts the domain of the function from $X_{\ge \sigma}(t)$ to $X_{\ge \pi}(t)$.
\end{lemma}
\begin{proof}
  We first show the injectivity of $\iota_\sigma$ through induction.
  For $\sigma$ with dimension $i = t$, the result is straightforward.
  For $\sigma$ with dimension $i = t-1$,
    the injectivity is implied from the identity axiom on $\sigma$.

  Assuming the statement holds for cells with dimension greater than $i$,
    we want to show that the statement also holds for $\sigma$ with dimension $i$.
  By the identity and gluability axioms on $\sigma$,
    we have
  \begin{equation}
    \cF_\sigma \stackrel{g_\sigma}{\cong} \{\{c_{\pi}\}_{\pi \in X_{\ge \sigma}(i+1)}:
      \forall \pi, \pi' \in X_{\ge \sigma}(i+1), \rho \in X_{\ge \sigma}(i+2),\,\,
      c_{\pi} \in \cF_{\pi},
      \res_{\pi,\rho}(c_{\pi}) = \res_{\pi',\rho}(c_{\pi'})\}
  \end{equation}
  where $c \mapsto \{\res_{\sigma,\pi}(c)\}_{\pi \in X_{\ge \sigma}(i+1)}$.
  The condition $\res_{\pi,\rho}(c_{\pi}) = \res_{\pi',\rho}(c_{\pi'})$
    is imposed only when $\pi \precdot \rho$ and $\pi' \precdot \rho$.
  Denote the RHS as $S$.

  To show that $\iota_\sigma$ is injective,
  we construct the inverse map $\im(\iota_\sigma) \to \cF_\sigma$.
  Given $h \in \im(\iota_\sigma) \subseteq \prod_{\tau \in X_{\ge \sigma}(t)} \cF_\tau$.
  Because $h(\tau) = \res_{\pi,\tau}(\res_{\sigma,\pi}(c))$ for all $\tau \succeq \pi$,
    we have $\prod_{\tau \in X_{\ge \pi}(t)} h(\tau) \in \iota_\pi(\cF_\pi)$.
  By the induction hypothesis, $\iota_\pi$ is injective,
    which allows us to define
    $c_\pi = \iota_\pi^{-1}(\prod_{\tau \in X_{\ge \pi}(t)} h(\tau)) \in \cF_{\pi}$
    for every $\pi \in X_{\ge \sigma}(i+1)$.
  It is straightforward to check that $\{c_{\pi}\}_{\pi \in X_{\ge \sigma}(i+1)} \in S$.
  Therefore, we can define the candidate inverse map $\iota_\sigma^{-1}$
    as $\iota_\sigma^{-1}(h) = g_\sigma^{-1}(\{c_{\pi}\}_{\pi \in X_{\ge \sigma}(i+1)})$.
  It is straightforward to check that the composition $\iota_\sigma^{-1} \circ \iota_\sigma$
    is the identity map on $\cF_\sigma$.

  The second statement about the restriction map
    is a direct consequence of the presheaf condition,
    $\res_{\pi,\tau} \res_{\sigma,\pi} = \res_{\sigma, \tau}$
    for $\tau \in X(t)$.
\end{proof}

Therefore, from now on,
  we will ocacasionally abuse the notation
  by viewing $\cF_\sigma$ as a subset of $\prod_{\tau \in X_{\ge \sigma}(t)} \cF_{\tau}$
  and interpreting elements of $\cF_\sigma$ as functions on $X_{\ge \sigma}(t)$.

\begin{lemma}\label{lem:structure-theorem-2}
  For every cell $\sigma \in X$,
  \begin{equation}\label{eq:structure-theorem}
    \cF_\sigma = \{c \in \prod_{\tau \in X_{\ge \sigma}(t)} \cF_{\tau}:
    \forall \tau' \in X_{\ge \sigma}(t-1), \,\, \res_{\sigma,\tau'}(c) \in \cF_{\tau'}\}.
  \end{equation}
  In particular, the sheaf $\cF$ is completely determined by
    $\{\cF_{\tau'}\}_{\tau' \in X(t-1)}$.
\end{lemma}
\begin{proof}
  We again proof by induction.
  For $\sigma$ with dimension $i = t-1$ or $t$ the statement is clearly true.

  Assuming the statement holds for cells with dimension greater than $i$,
    we want to show that the statement also holds for $\sigma$ with dimension $i$.
  Denote the RHS of \Cref{eq:structure-theorem} as $T_\sigma$.
  By the identity and gluability axioms on $\sigma$
  \begin{equation}
    \cF_\sigma \stackrel{g_\sigma}{\cong} \{\{c_{\pi}\}_{\pi \in X_{\ge \sigma}(i+1)}:
      \forall \pi, \pi' \in X_{\ge \sigma}(i+1), \rho \in X_{\ge \sigma}(i+2),\,\,
      c_{\pi} \in \cF_{\pi},
      \res_{\pi,\rho}(c_{\pi}) = \res_{\pi',\rho}(c_{\pi'})\}.
  \end{equation}
  Denote the RHS as $S$.

  It is clear that $\cF_\sigma \subseteq T_\sigma$,
    which is a direct consequence of the presheaf condition.
  What remains is to show that $T_\sigma \subseteq \cF_\sigma$.
  Given $h \in T_\sigma$,
    we define $c_\pi = \prod_{\tau \in X_{\ge \pi}(t)} h(\tau)$
    for every $\pi \in X_{\ge \sigma}(i+1)$.
  It is clear that $c_\pi \in T_\pi$.
  By the induction hypothesis, $T_\pi = \cF_\pi$.
  Therefore, $\{c_{\pi}\}_{\pi \in X_{\ge \sigma}(i+1)} \in S$,
    which implies $h \in \cF_\sigma$,
    as desired.
\end{proof}

In fact, the statements of \Cref{lem:structure-theorem-1,lem:structure-theorem-2}
  together is equivalent to the identity and gluability axioms.
The means that we could have defined sheaves over cellular complexes
  using the statements in the lemmas.
We denote the unique sheaf with $\iota_\sigma(\cF_{\sigma}) = C_{\sigma} \subseteq \FF_q^{X_{\ge \sigma}(t)}$
  as $\cF(X, \{C_{\sigma}\}_{{\sigma} \in X(t-1)})$,
We refer to $C_{\sigma}$ as the local codes,
  which serves a role similar to the local codes used in Tanner codes.

Conceptually, these $\cF_{\tau'}$ for $\tau' \in X(t-1)$
  should be regarded as analogs of stalks.
In the context of sheaf over topological spaces,
  it is known that
  ``many properties of the sheaf are determined by its stalks''.
Similarly, in our context of cellular complex,
  we can adopt the same principle:
  ``many properties of the global code are determined by its local codes''.

This characterization of sheaves allows us to define the concept of a dual sheaf $\cF^\perp$,
  where $\iota_\sigma((\cF^\perp)_\sigma) = (\iota_\sigma(\cF_\sigma))^\perp$
  for $\sigma \in X(t-1)$.
  (This notion was also proposed in \cite{panteleev2024maximally} using the notation $\cF^T$.)
As we will see,
  under suitable conditions,
  there exists a Poincaré duality between $\cF$ and $\cF^\perp$.
This will allow us to identify the X and Z logical operators with the cohomologies of $\cF$ and $\cF^\perp$,
  respectively.

Interestingly, the two statements in the lemmas
  were used in \cite{panteleev2024maximally} to define the Tanner sheaves.
The only technical difference
  is that we allow $\cF_\tau \ne \FF_q$ for $\tau \in X(t)$.
In \cite{panteleev2024maximally},
  the statement of \Cref{lem:structure-theorem-1} is assumed as part of the definition of sheaves,
  and the statement of \Cref{lem:structure-theorem-2} is treated as an additional requirement specific for Tanner sheaves.
In contrast,
  we argue that these two conditions probably should be included right from the beginning in the definition of sheaves over cellular complexes,
  as they are closely related to the identity and gluability axioms that are central to the definition of sheaves over topological spaces.

\begin{remark}
  Our definition of sheaf depends on the value of $t$.
  In particular, we apply different conditions
    for $i = t-1$ and $t \le t-2$.
  As a result, if we take the discrete union of
    a $t$-dimensional cellular complex $X$ with sheaf $\cF$
    and a $t'$-dimensional cellular complex $X'$ with sheaf $\cF'$,
    the resulting structure is generally not a sheaf when $t \ne t'$,
    even though the discrete union of $X$ and $X'$ reamins a cellular complex.
  This contrasts with the definition of sheaves in \cite{curry2014sheaves}

  We argue that this behavior is a feature, not a bug.
  In our context,
    even though we describe the objects through $t$-dimensional cellular complexes,
    the goal is to model $t$-dimensional manifolds.
  Similarly, the discrete union
    of two manifolds of different dimensions is no longer a manifold.
  In particular, this dependence on $t$ is manifest in the statement of Poincaré duality,
    which we will establish in \Cref{sec:poincare}.
\end{remark}

\subsection{Sheaf cohomology and chain complexes}

Given a sheaf $\cF$ on a cellular complex $X$, we can construct a chain complex $\cC(X, \cF)$
  analogous to the construction of the cellular cohomology.
This sheaf-induced chain complex has recently become
  a valuable tool for constructing classical and quantum codes
  \cite{first2022good,dinur2023new,dinur2024expansion,panteleev2024maximally}.

We first define the vector spaces.
For $0 \le i \le t$,
  let $\cC^i(X, \cF)$ be the space of cochains defined on $\sigma \in X(i)$
  with values in $\cF_\sigma$
\begin{equation}
  \cC^i(X, \cF) = \bigoplus_{\sigma \in X(i)} \cF_\sigma.
\end{equation}
Otherwise, $\cC^i(X, \cF) = 0$.
It is clear that $\cC^i(X, \cF)$ forms a vector space, since each $\cF_\sigma$ is a vector space.

We now define the coboundary maps $\delta^i: \cC^i(X, \cF) \to \cC^{i+1}(X, \cF)$.
Let $\a \in C^i$ be a cochain.
To define $\delta^i(\a) \in \cC^{i+1}(X, \cF)$,
  it is sufficient to specify $(\delta^i \a)(\tau) \in \cF_\tau$ for every $\tau \in X(i+1)$.
In particular, for $0 \le i \le t-1$ we define
\begin{equation}
  \delta^i \a (\tau) = \sum_{\sigma \precdot \tau} \res_{\sigma, \tau}(\a(\sigma)).
\end{equation}
Otherwise, $\delta^i$ is defined as the zero map.
(In general, the definition contains signs, $\sum_{\sigma \precdot \tau} (-)^{(...)} \res_{\sigma, \tau}(\a(\sigma))$.
  We evade this technical discussion by studying finite fields with characteristic $2$
  where $-1 = 1$.)

We finally check that $\delta \circ \delta = 0$. We have
\begin{equation}
  \delta^{i+1} \delta^i \a(\pi) = \sum_{\sigma \precdot \tau \precdot \pi} \res_{\sigma,\pi}(\a(\sigma)) = 0.
\end{equation}
The last equality holds
  because for each pair $\sigma \in X(i), \pi \in X(i+2)$,
  there exist two $\tau \in X(i+1)$ such that $\sigma \precdot \tau \precdot \pi$,
  as discussed near \Cref{claim:two}.

Note that the standard cellular complex of $X$
  can also be viewed as a chain complex induced from a sheaf,
  where the sheaf consists of constant sections,
  i.e. each $\cF_\sigma \cong \FF_q$ forms a repetition code.
So it makes sense to overload the notation
  and express the cellular complex as $\cC(X, \FF_q)$.

Using the chain complex $\cC^{\bullet}(X, \cF)$,
  we can define the dual chain complex $\cC_{\bullet}(X, \cF)$.
These two chain complexes
  further induce the cohomology and homology groups, $H^\bullet(X, \cF)$ and $H_\bullet(X, \cF)$.

We observe a relation between the homology and cohomology of the sheaf and its dual sheaf.
\begin{claim}\label{claim:poincare-0}
  Given a sheaf $\cF$ on a cellular complex $X$, we have
  \begin{equation}
    H_t(X, \cF) \cong H^0(X, \cF^\perp).
  \end{equation}
\end{claim}
\begin{proof}
  Following from the characterization of sheaves,
    let $\cF = \cF(X, \{C_{\sigma}\}_{{\sigma} \in X(t-1)})$
    and $\cF^\perp = \cF(X, \{C_{\sigma}^\perp\}_{{\sigma} \in X(t-1)})$.
  By definition, $H_t(X, \cF)$ is induced from the following part of the chain complex
    $0 \xrightarrow{\partial_{t+1}} \cC_{t}(X, \cF) \xrightarrow{\partial_t} \cC_{t-1}(X, \cF)$
  and
  $H^0(X, \cF^\perp)$ is induced from the following part of the chain complex
    $0 \xrightarrow{\delta^{-1}} \cC^0(X, \cF^\perp) \xrightarrow{\delta^{0}} \cC^1(X, \cF^\perp)$.
  That means
  \begin{equation}
    H_t(X, \cF) = \ker \partial_t = \{\a \in \prod_{\tau \in X(t)} \cF_\tau:
      \forall \sigma \in X(t-1), \sum_{\tau \succdot \sigma} \res_{\sigma, \tau}^T(\a(\tau)) = 0\}
  \end{equation}
  and
  \begin{equation}
    H^0(X, \cF^\perp) = \ker \delta^0 = \{\beta \in \prod_{\pi \in X(0)} \cF_\pi:
      \forall \rho \in X(1), \sum_{\pi \precdot \rho} \res_{\pi, \rho}(\beta(\pi)) = 0\}.
  \end{equation}
  Recall that, in our context $\cF_\sigma$ has a choice of basis,
    which allows us to define the transpose $\res_{\sigma, \tau}^T: \cF_\tau \to \cF_\sigma$.
  We would like to show that both are isomorphic to the set of global sections of $\cF^\perp$
  \begin{equation}
    S = \{\a \in \prod_{\tau \in X(t)} \cF_\tau: \forall \sigma \in X(t-1), \{\a(\tau)\}_{\tau \in X_{\ge \sigma}(t)} \in C_\sigma^\perp\}.
  \end{equation}

  We begin by showing $H_t(X, \cF) \cong S$.
  The key is to understand the meaning of
    $\sum_{\tau \succdot \sigma} \res_{\sigma, \tau}^T(\a(\tau)) = 0$.
  Let $\iota_{\sigma}: \cF_\sigma \to \prod_{\tau \in X_{\ge \sigma}(t)} \cF_\tau$
    be the map such that $c \mapsto \prod_{\tau \in X_{\ge \sigma}(t)} \res_{\sigma, \tau}(c)$.
  Notice that
  \begin{equation}
    \sum_{\tau \succdot \sigma} \res_{\sigma, \tau}^T(\alpha(\tau)) = 0
    \iff
    \iota_{\sigma}^T(\{\alpha(\tau)\}_{\tau \in X_{\ge \sigma}(t)}) = 0
    \iff
    \{\alpha(\tau)\}_{\tau \in X_{\ge \sigma}(t)} \in C_\sigma^\perp.
  \end{equation}
  The last iff holds because $C_\sigma = \im \iota_\sigma$.

  We now show $H^0(X, \cF^\perp) \cong S$.
  For each $\rho \in X(1)$,
    there are exactly two $\pi_0, \pi_1 \in X(0)$ such that $\pi_0, \pi_1 \precdot \rho$.
  This implies that the restrictions on $\rho$ agree,
    $\res_{\pi_0, \rho}(\b(\pi_0)) = \res_{\pi_1, \rho}(\b(\pi_1))$.
  Similar to the proof in \Cref{lem:structure-theorem-2},
    one can show that for each $\tau \in X(t)$,
    the value $\res_{\pi, \tau}(\b(\pi))$ is independent of $\pi \in X(0)$.
  This allows us to map $\b$ to the function $\a \in \prod_{\tau \in X(t)} \cF_\tau$,
    where $\a(\tau) = \res_{\pi, \tau}(\b(\pi))$ for any choice of $\pi \in X(0)$.
  The image $\a$ satisfies
    $\{\a(\tau)\}_{\tau \in X_{\ge \sigma}(t)} \in C_\sigma^\perp$ for all $\sigma \in X(t-1)$
    because $\b(\pi) \in \cF^\perp_\pi$ and its restriction $\res_{\pi, \sigma} \in \cF^\perp_\sigma$.
  It is straightforward to check that this is an isomorphism.
  (This relation between $H^0(X, \cF^\perp)$ and the global sections holds generally for sheaves.)
\end{proof}

The equation $H_t(X, \cF) \cong H^0(X, \cF^\perp)$ is reminiscent of Poincaré duality.
In fact, as we will show in \Cref{sec:poincare},
  under suitable assumptions detailed in \Cref{sec:locally-acyclic-sheaves},
  this relation holds more generally,
  with $H_{t-i}(X, \cF) \cong H^i(X, \cF^\perp)$
  for $0 \le i \le t-1$.

\subsection{Quantum LDPC codes based on sheaves}

We are ready to describe a family of qLDPC codes
  based on Tanner sheaf complexes.
We will later show that they have a natural geometric interpretation
  and a natural cup product.

Recall that a quantum CSS code is equivalent to a three-term chain complex.
The previous section provides a chain complex
  $\cC(X, \{C_\sigma\}_{\sigma \in X(t-1)})$.
So the only additional data is an integer $\ell$
  which specifies which of the three terms to select.
This leads to the quantum code $Q(X, \{C_\sigma\}_{\sigma \in X(t-1)}, \ell)$
  defined by
\begin{equation}
  \cC^{\ell-1}(X, \{C_\sigma\}_{\sigma \in X(t-1)}) \to \cC^{\ell}(X, \{C_\sigma\}_{\sigma \in X(t-1)}) \to \cC^{\ell+1}(X, \{C_\sigma\}_{\sigma \in X(t-1)}).
\end{equation}
The X distance is related to the small set coboundary expansion \cite{hopkins2022explicit} of the above chain complex.

Following the discussions in \Cref{sec:interpret-logical-operators},
  and more rigorously in \Cref{sec:poincare}\footnote{
    The discussion in \Cref{sec:poincare} is not complete,
    as we have only show the correspondece at the level of homology.
    To establish the claimed result, we must also show a correspondence that includes a bound between the weight of the chain and cochain.
    Nevertheless, a similar setting for cubical complexes was addressed in \cite[Sec 7]{dinur2024expansion}
      and it is possible to generalize that proof to general cellular complexes.},
  the Z distance is related to the small set coboundary expansion of another chain complex
\begin{equation}
  \cC^{t-\ell-1}(X, \{C_\sigma^\perp\}_{\sigma \in X(t-1)}) \to \cC^{t-\ell}(X, \{C_\sigma^\perp\}_{\sigma \in X(t-1)}) \to \cC^{t-\ell+1}(X, \{C_\sigma^\perp\}_{\sigma \in X(t-1)}).
\end{equation}

We will not discuss how to prove small set coboundary expansion in this paper,
  but the tools in \cite[Sec 6]{dinur2024expansion} are well-suited for this task.

\subsection{A map from cochains to chains for sheaf complexes}

Similar to \Cref{sec:map-logical-to-geometry},
  we map Pauli X operators to geometric objects
  which can be used to define the triple intersection number.
More specifically, we will map cochains to chains.
Strictly speaking, this discussion is not necessary,
  since we will take a different route and define the cup product directly.
However,
  we still provide this discussion
  because these geometric objects are easier to visualized
  and are helpful to keep in mind.
For simplicity, in this and the following sections,
  we will primarily focus on the case where $\cF_{\tau} = \FF_q$
  for all $\tau \in X(t)$.
In \Cref{sec:abstract}, we will extend the discussion to the case where $\cF_{\tau}$ is not necessarily $\FF_q$.

For $0 \le i \le t$,
  we will map every cochain
  $\a \in \cC^{i}(X, \cF)$,
  to a chain on the ``interior'' of the nerve of $X$,
  denoted as $\tilde \a \in \cC_{t-i}(\cN_{int}(X), \FF_q)$.
The nerve of $X$, denoted as $\cN(X)$, is a simplicial complex
  where the vertices correspond to the cells in $X$,
  and the edges correspond to pairs of cells related by inclusion
  $\{(\sigma, \tau): \sigma \precneqq \tau\}$.
More generally,
  the $i$-simplices of $\cN(X)$ correspond to sequences of cell inclusions of length $i+1$,
  $\{(\sigma_0, ..., \sigma_i): \sigma_0 \precneqq ... \precneqq \sigma_i\}$.
The boundary of the nerve $\cN_{\partial}(X)$
  consists of simplices
  that lie on the $(t-1)$-skeleton of $X$.
The interior of the nerve $\cN_{int}(X)$
  consists of simplices in $\cN(X) \backslash \cN_{\partial}(X)$.
Formally,
  the chains in $\cC_{\bullet}(\cN_{int}(X), \FF_q)$
  can be thought of as the quotient $\cC_{\bullet}(\cN(X), \FF_q)/\cC_{\bullet}(\cN_{\partial}(X), \FF_q)$,
  as in the construction of relative homology.

We are ready to define the map from a cochain $\a \in \cC^{i}(X, \cF)$
  to a chain $\tilde \a \in \cC_{t-i}(\cN_{int}(X), \FF_q)$.
Given a pair of cells $\sigma \in X(i), \tau \in X(t)$
  with $\sigma \preceq \tau$.
Let $S(\sigma, \tau)$ be the collection of $(t-i)$-simplices,
  $\{(\sigma_0, ..., \sigma_{t-i}): \sigma_0 \precneqq ... \precneqq \sigma_{t-i}\}$
  with $\sigma_0 = \sigma$ and $\sigma_{t-i} = \tau$.
We assign the value $\res_{\sigma,\tau}(\a(\sigma)) \in \FF_q$ to every simplex in $S(\sigma, \tau)$.
Formally,
\begin{equation}
  \tilde \a = \sum_{\sigma \in X(i), \tau \in X(t)} \sum_{\pi \in S(\sigma, \tau)} \res_{\sigma,\tau}(\a(\sigma)) \pi.
\end{equation}
Notice that $\tilde \a$ do not involve any simplices that lie on the boundary $\cN_{\partial}(X)$.

This construction is consistent with the one used for cubical complexes discussed in \Cref{sec:map-logical-to-geometry}.
Similar to the case of a cubical complex,
  the corresponding $(t-i)$-dimensional surface $\tilde \a$
  ends on the boundaries in a specific way,
  where the values near a boundary $\tau' \in X(t-1)$ form a vector in $C_{\tau'} \subseteq \FF_q^{X_{\ge \tau'}(t)}$.

This geometric perspective helps clarify the upcoming discussion of the cup product.
The dual of the cup product is the intersection.
Consider a $(t-i)$-dimensional surface $\tilde \a_1$ associated with $\a_1 \in \cC^{i}(X, \{C_{1,\tau'}\}_{\tau' \in X(t-1)})$
  and a $(t-j)$-dimensional surface $\tilde \a_2$ associated with $\a_2 \in \cC^{j}(X, \{C_{2,\tau'}\}_{\tau' \in X(t-1)})$.
The intersection of $\tilde \a_1$ and $\tilde \a_2$
  is a $(t-i-j)$-dimensional surface $\tilde \a$.
Furthermore, the values near the boundary of $\tilde \a$
  are derived from the product of the boundary values of $\tilde \a_1$ and $\tilde \a_2$,
  which lie in $C_{1,\tau'} \odot C_{2,\tau'}$,
  where $C_1 \odot C_2 = \{c_1 \odot c_2: c_1 \in C_1, c_2 \in C_2\}$
  and $c_1 \odot c_2$ is the entrywise product of $c_1$ and $c_2$.
By taking the dual from the surface $\tilde c$ back to the cochain $c$,
  this suggests that the the cup product of $c_1$ and $c_2$
  should lie in $\cC^{i+j}(X, \{C_{1,\tau'} \odot C_{2,\tau'}\}_{\tau' \in X(t-1)})$.
This is exactly what we will show in the remaining sections.

\subsection{Cup product on a sheaf complex over a simplicial complex}

We first consider a simpler case where $X$ is a simplicial complex.
In the next section, we will triangulate general cellular complex,
  which allows us to reduce the problem to the case of simplicial complex.

The goal is to define a cup product for the cohomology classes
\begin{equation}
  H^i(X, \cF_1) \times H^j(X, \cF_2)
  \xrightarrow{\cup} H^{i+j}(X, \cF_1 \odot \cF_2)
\end{equation}
where $\cF_1 \odot \cF_2 = \cF(X, \{C_{1,\sigma} \odot C_{2,\sigma}\}_{\sigma \in X(t-1)})$
where
  $C_1 \odot C_2 = \Span\{c_1 \odot c_2: c_1 \in C_1, c_2 \in C_2\}$
  is the vector space spanned by $c_1 \odot c_2$
  and $c_1 \odot c_2$ is the entrywise product of $c_1, c_2$.
To do so, we will first define a cup product for the cochains
\begin{equation}
  \cC^i(X, \cF_1) \times \cC^j(X, \cF_2)
  \xrightarrow{\cup} \cC^{i+j}(X, \cF_1 \odot \cF_2)
\end{equation}
and show that such map interacts well with the coboundary map.

Notice that the ``presheaf'' constructed from
  $\{\cF_{1,\sigma} \odot \cF_{2,\sigma}\}_{\sigma \in X}$
  is generally different from the sheaf $\cF_1 \odot \cF_2$ defined earlier.
In particular, while $\cF_{1,\sigma} \odot \cF_{2,\sigma} \subseteq (\cF_1 \odot \cF_2)_\sigma$,
  they are not necessarily equal.
This parallels the situation the traditional sheaf theory,
  where one often needs to apply sheafification after performing elementary operations.

Recall that for two cochains on a simplicial complex
  $\a \in \cC^i(X, \FF_q), \b \in \cC^j(X, \FF_q)$,
  the cup product $\a \cup \b \in \cC^{i+j}(X, \FF_q)$ is defined by
  \begin{equation}
    (\a \cup \b)([v_0, ..., v_i, ..., v_{i+j}]) = \a([v_0, ..., v_i]) \b([v_i, ..., v_{i+j}])
  \end{equation}
  for every $(i+j)$-simplex $[v_0, ..., v_i, ..., v_{i+j}]$ in $X$,
  where $[v_0, ..., v_i]$ denotes the simplex with vertices $v_0, ..., v_i$.
Notice that the definition assumes an ordering of the vertices.

The idea behind defining the new cup product is simple.
  We combine the contributions from each simplex
  analogous to how we combine the values from each intersection.
Since the cup product on each simplex depends on the vertex ordering,
  we assign a global ordering to all vertices in $X$ from the start.
This ensures a consistent ordering among all simplices,
  which will be important later on.

Given $\a \in \cC^i(X, \cF)$,
  let $\a_{\le \tau}: X_{\le \tau}(i) \to \FF_q$
  be the map where $\sigma \mapsto \a(\sigma)|_\tau$.
(Recall that $c|_\tau$ is a shorthand for $\res_{\sigma,\tau} (c)$)
This local component $\a_{\le \tau}$ captures the $(t-i)$-dimensional surface structure
  in the $t$-simplex $\tau$.
It is clear that there is a bijection between $\a \in \cC^i(X, \cF^{full})$ and $\{\a_{\le \tau}\}_{\tau \in X(t)} \in \prod_{\tau \in X(t)} (X_{\le \tau}(i) \to \FF_q)$
  where $\cF^{full}_{\sigma} = (X_{\ge \sigma}(t) \to \FF_q)$ could be any local section.
That means given the values $\{\a_{\le \tau}\}_{\tau \in X(t)}$,
  we can combine them into an element in $\cC^i(X, \cF^{full})$.

Notice that $\a_{\le \tau}: X_{\le \tau}(i) \to \FF_q$
  can be interpreted as a $i$-cochain on the simplex $X_{\le \tau}$,
  since $(X_{\le \tau}(i) \to \FF_q) \cong \cC^i(X_{\le \tau}, \FF_q)$.
Thus, we define $\a \cup \b$ as
\begin{equation}
  (\a \cup \b)_{\le \tau} = (\a_{\le \tau}) \cup (\b_{\le \tau}),
\end{equation}
where $(\a_{\le \tau}) \cup (\b_{\le \tau})$ is the usual cup product
  $\cC^i(X_{\le \tau}, \FF_q) \times \cC^j(X_{\le \tau}, \FF_q)
  \xrightarrow{\cup} \cC^{i+j}(X_{\le \tau}, \FF_q)$.

At this moment, we only know that $\a \cup \b \in \cC^{i+j}(X, \cF^{full})$.
We now show that $\a \cup \b \in \cC^{i+j}(X, \cF_1 \odot \cF_2)$.
All we need to show is that $(\a \cup \b)|_\sigma \in \cF_{1,\sigma} \odot \cF_{2,\sigma}$, since $\cF_{1,\sigma} \odot \cF_{2,\sigma} \subseteq (\cF_1 \odot \cF_2)_\sigma$.
We have
\begin{multline}
  (\a \cup \b)(\sigma)
  = \{(\a \cup \b)(\sigma)|_\tau\}_{\tau \in X_{\ge \sigma}(t)}
  = \{(\a \cup \b)_{\le \tau}(\sigma)\}_{\tau \in X_{\ge \sigma}(t)}
  = \{((\a_{\le \tau}) \cup (\b_{\le \tau}))(\sigma)\}_{\tau \in X_{\ge \sigma}(t)} \\
  = \{\a_{\le \tau}(\sigma_1) \b_{\le \tau}(\sigma_2)\}_{\tau \in X_{\ge \sigma}(t)}
  = \{\a_{\le \tau}(\sigma_1)\}_{\tau \in X_{\ge \sigma}(t)} \odot \{\b_{\le \tau}(\sigma_2)\}_{\tau \in X_{\ge \sigma}(t)}
\end{multline}
where the equalities in the first line follow from definitions.
The first equality in the second line follows from the definition of the usual cup product.
For $\sigma = [v_0, ..., v_i, ..., v_{i+j}]$,
  we define $\sigma_1 = [v_0, ..., v_i], \sigma_2 = [v_i, ..., v_{i+j}]$.
Note that this step relies on having a global ordering of the vertices,
  which ensures the vertex ordering is independent of $\tau$.

Notice that $\{\a_{\le \tau}(\sigma_1)\}_{\tau \in X_{\ge \sigma}(t)}$
  is a restriction of $\{\a_{\le \tau}(\sigma_1)\}_{\tau \in X_{\ge \sigma_1}(t)} = \a(\sigma_1)$.
Because $\a(\sigma_1) \in \cF_{1, \sigma_1}$ and by the definition of the sheaf,
  the restriction to $X_{\ge \sigma}(t)$ is a vector in $\cF_{1, \sigma}$,
  i.e. $\{\a_{\le \tau}(\sigma_1)\}_{\tau \in X_{\ge \sigma}(t)} \in \cF_{1, \sigma}$.
Similarly, $\{\b_{\le \tau}(\sigma_2)\}_{\tau \in X_{\ge \sigma}(t)} \in \cF_{2, \sigma}$.
Therefore, $(\a \cup \b)(\sigma) \in \cF_{1, \sigma} \odot \cF_{2, \sigma}$.

We have established a cup product for the cochains
  $\cC^i(X, \cF_1) \times \cC^j(X, \cF_2)
  \xrightarrow{\cup} \cC^{i+j}(X, \cF_1 \odot \cF_2)$.
We now show that it interacts well with the coboundary operator
  and induces a cup product for the cohomology classes.

First, a simple lemma.
\begin{lemma} \label{lem:delta-restrict}
  $(\delta \a)_{\le \tau} = \delta (\a_{\le \tau})$
  for $\a \in \cC^i(X, \cF)$.
\end{lemma}

\begin{proof}
  We simply apply the definitions
  \begin{equation}
    (\delta \a)_{\le \tau}(\sigma)
    = (\delta \a)(\sigma)|_\tau
    = \sum_{\sigma' \precdot \sigma} \a(\sigma')|_\tau
    = \sum_{\sigma' \precdot \sigma} \a_{\le \tau}(\sigma')
    = \delta (\a_{\le \tau})
  \end{equation}
  where the first and third equalities follow from the definition of $(*)_{\le \tau}$.
  The second equality follows from the definition of $\delta$ on $\cC^i(X, \cF)$
    and the forth equality follows from the definition of $\delta$ on $\cC^i(X_{\le \tau}, \FF_q)$.
\end{proof}

This is the key formula.
\begin{lemma}
  $\delta(\a \cup \b) = (\delta \a \cup \b) + (\a \cup \delta \b)$
  for $\a \in \cC^i(X, \cF_1)$ and $\b \in \cC^j(X, \cF_2)$.
\end{lemma}

\begin{proof}
  All we need to show is that
  $(\delta(\a \cup \b))_{\le \tau} = (\delta \a \cup \b)_{\le \tau} + (\a \cup \delta \b)_{\le \tau}$ for all $\tau \in X(t)$.

  We have
  \begin{multline}
    (\delta(\a \cup \b))_{\le \tau}
    = \delta((\a \cup \b)_{\le \tau})
    = \delta(\a_{\le \tau} \cup \b_{\le \tau})
    = (\delta (\a_{\le \tau}) \cup \b_{\le \tau}) + (\a_{\le \tau} \cup \delta (\b_{\le \tau})) \\
    = ((\delta \a)_{\le \tau} \cup \b_{\le \tau}) + (\a_{\le \tau} \cup (\delta \b)_{\le \tau})
    = (\delta \a \cup \b)_{\le \tau} + (\a \cup \delta \b)_{\le \tau}
  \end{multline}
  where the first and fifth equalities follow from \Cref{lem:delta-restrict},
    and the second and forth equalities follow from the definition of the new cup product.
  The third equality is the nontrivial result that holds for the usual cup product (see, for example, \cite[Lemma 3.6]{hatcher2000algebraic}).
\end{proof}

It is clear that if $\a$ and $\b$ are cocycles, then $\a \cup \b$ is again a cocycle.
Also, if one of $\a, \b$ is a cocycle and the other is a coboundary,
  say $\delta \a = 0$ and $\b = \delta \c$,
  then $\a \cup \b = \a \cup \delta \c = \delta(\a \cup \c)$ is again a coboundary.
This implies that there is an induced cup product on the cohomology classes
$H^i(X, \cF_1) \times H^j(X, \cF_2)
  \xrightarrow{\cup} H^{i+j}(X, \cF_1 \odot \cF_2)$.

\subsection{Cup product on a sheaf complex}

For the purpose of constructing quantum code,
  we can actually stop here
  and use the result from the last section on simplicial complexes.
The reason is that for any cellular complex $X$ with sheaf $\cF$,
  we can always subdivide it into a simplicial complex $\tilde X$
  and extend the sheaf into $\tilde \cF$ (see below).
The codes defined on $(X, \cF)$ and $(\tilde X, \tilde \cF)$
  are essentially equivalent.
Therefore, the result from the last section is sufficient for all practical purpose.
Nevertheless, since our discussion in \Cref{sec:CCZ-cubical-code} was on cubical complex,
  we will use this opportunity to demonstrate the tools to convert between simplicial and cellular complexes.

Given a cellular complex $X$ and a sheaf $\cF$ on $X$.
Let $\tilde X$ be a simplicial complex that triangulates $X$.
The goal is to use the cup products defined on
  $\cC(\tilde X, \tilde \cF)$ and $H(\tilde X, \tilde \cF)$ in the previous section
  to induce cup products on $\cC(X, \cF)$ and $H(X, \cF)$.
The idea fairly simple:
  since $X$ and $\tilde X$ have the same topology,
  the chain complexes defined on $X$ and $\tilde X$
  will be chain homotopy equivalent.
However, explicitly constructing the maps between them is somewhat involved,
  and will form the bulk of this section.

\paragraph*{Define the sheaf $\tilde \cF$:}

We begin by defining $\tilde \cF$.
To achieve this,
  we need to specify a vector space $\tilde \cF_{\sigma} \subseteq \FF_q^{\tilde X_{\ge \sigma}(t)}$
  for every $\sigma \in \tilde X$.
The idea is that each $t$-cell in $X$ is decomposed into many $t$-simplices in $\tilde X$.
The vector space $\tilde \cF_\sigma$
  is induced from such a decomposition
  by duplicating the value in a $t$-cell to the $t$-simplices it contains.
Let $\tau_\sigma$ be the smallest cell in $X$ that contains $\sigma$.
Notice that every $t$-simplices above $\sigma$
  is contained in some $t$-cells above $\tau_\sigma$.
Let $I: \tilde X_{\ge \sigma}(t) \to X_{\ge \tau_\sigma}(t)$
  be the map that maps
  every $t$-simplex to the unique $t$-cell it is contained in.
This induces the map $I_*: (X_{\ge \tau_\sigma}(t) \to \FF_q) \to (\tilde X_{\ge \sigma}(t) \to \FF_q)$,
  which allows us to define $\tilde \cF_\sigma$ as the image $I_*(\cF_{\tau_\sigma})$.
It is not hard to check that $\cF_{\tau_\sigma} \cong \tilde \cF_\sigma$.
Essentially, the map duplicates the values in $X_{\ge \tau_\sigma}(t)$
  to the values $\tilde X_{\ge \sigma}(t)$.
It is also not hard to check that this duplication operation commutes with
  the entrywise product,
  $\widetilde{\cF_1 \odot \cF_2} = \tilde{\cF_1} \odot \tilde{\cF_2}$.

\paragraph*{Define the maps $g$ and $\tilde g$ between the topological spaces $X$ and $\tilde X$:}

Let $X_i$ and $\tilde X_i$ be the $i$-skeleton of $X$ and $\tilde X$.
We now view $X$ and $\tilde X$ as topological spaces
  and construct two maps $g: X \to \tilde X$ and $\tilde g: \tilde X \to X$,
  such that $\tilde g g$ and $g \tilde g$ are homotopic to the identity maps on $X$ and $\tilde X$, respectively.
Furthermore,
  we require that the $i$-skeleton is map into the $i$-skeleton,
  $g(X_i) \subseteq \tilde X_i, \tilde g(\tilde X_i) \subseteq X_i$.
This allows us to use $g$ and $\tilde g$ to naturally induce
  the maps between the chains (cochains) of $X$ and $\tilde X$.
Additionally,
  for each simplex $\sigma \in \tilde X$,
  we want $\tilde g(\sigma) \subseteq \tau_\sigma$
  where $\tau_\sigma \in X$ is the smallest cell in $X$ that contains $\sigma$.
This guarantees the correspondence between the local sections.

Since $X_i \subseteq \tilde X_i$, we can naturally define $g$ as such an embedding.
Defining $\tilde g$ is more complicated
  and we will construct it iteratively.
We begin by defining $\tilde g$ on the $0$-cells.
For each $\sigma \in \tilde X(0)$,
  we map $\sigma$ to a vertex in $\tau_\sigma$.
Now, assume that $\tilde g: \tilde X_i \to X_i$ is defined.
Our goal is to extend the definition of $\tilde g$ to every $(i+1)$-simplex $\sigma \in \tilde X(i+1)$.
Given the image on the boundary $\tilde g(\partial \sigma)$,
  we want to ensure that $\tilde g(\sigma)$ lies within $\tau_\sigma$.
By the induction hypothesis,
  it is known that $\tilde g(\partial \sigma)$ is supported on $\tau_\sigma$.
Since $\tau$ is contractable,
  $\tilde g(\partial \sigma)$ is contractable on the $(i+1)$-skeleton of $\tau_{\sigma}$, $(\tau_{\sigma})_{i+1}$.
Such a contraction induces a map from $\sigma$ into $(\tau_{\sigma})_{i+1}$,
  thereby defining $\tilde g$ on $\sigma$.
By continuing this process, we obtain $\tilde g: \tilde X \to X$.

We check that $\tilde g g$ and $g \tilde g$ are homotopic to the identity maps.
It is clear that $\tilde gg$ is the the identity map on $X$.
For $g \tilde g$, one can again build the homotopy map iteratively,
  which we will skip the details.

\paragraph*{Define the induced chain and cochain maps:}

With $g$ and $\tilde g$, we can define the corresponding chain maps
  $g_i: \cC_i(X, \ZZ) \to \cC_i(\tilde X, \ZZ)$
  and $\tilde g_i: \cC_i(\tilde X, \ZZ) \to \cC_i(X, \ZZ)$.
Notice that, by construction,
  the images of $g$ and $\tilde g$ for each cell is a combination of cells,
  which leads to the map on the chains.
In particular, $g_i$ can be explicitly expressed as,
  $g_i(\sigma) = \sum_{\tilde \sigma \in \tilde X(i), \tilde \sigma \subseteq \sigma} \tilde \sigma$,
  where $\tilde \sigma$ are the simplices of the same dimension that form the cell $\sigma$.
For $\tilde g_i(\tilde \sigma)$, it is harder to express it explicitly,
  but we have the property that $\tilde g_i(\tilde \sigma)$ is only supported on the chain within $\tau_{\tilde \sigma}$.
Since these chain maps are induced from maps on the topological spaces,
  they naturally commute with the boundary operators.

The chain maps $g_i, \tilde g_i$
  naturally induce the maps on the cochains
  $g^i: \cC^i(\tilde X, \tilde \cF) \to \cC^i(X, \cF)$,
  $\tilde g^i: \cC^i(X, \cF) \to \cC^i(\tilde X, \tilde \cF)$.
For $g^i$, we define
\begin{equation}
  g^i (\tilde \a)(\sigma) = \sum_{\tilde \sigma \in g_i(\sigma)} I_*^{-1}(\tilde \a(\tilde \sigma)).
\end{equation}
Recall $I_*: \cF_{\tau_{\tilde \sigma}} \to \cF_{\tilde \sigma}$ is an isomorphism.
Because $\tilde \sigma \subseteq \sigma$ and both have the same dimension,
  we have $\tau_{\tilde \sigma} = \sigma$.
Therefore, $I_*^{-1}(\tilde \a(\tilde \sigma)) \in \cF_\sigma$
  and the expression is well-defined.
For $\tilde g^i$, we define
\begin{equation}
  \qquad \tilde g^i (\a)(\tilde \sigma) = I_*\Big(\sum_{\sigma \in \tilde g_i(\tilde \sigma)} \a(\sigma)|_{\tau_{\tilde \sigma}}\Big).
\end{equation}
By the construction of $\tilde g$,
  $\sigma$ is a cell in $\tau_{\tilde \sigma}$.
This allows us to restrict $\alpha(\sigma) \in \cF_{\sigma}$
  to $\alpha(\sigma)|_{\tau_{\tilde \sigma}} \in \cF_{\tau_{\tilde \sigma}}$.
Note that the sum accounts for the multiplicities:
  if $\sigma$ is covered twice in $\tilde g_i(\tilde \sigma)$,
  then we consider two contributions from $\sigma$.
Since these cochain maps are induced from maps on the topological spaces,
  they naturally commute with the coboundary operators.

\paragraph*{Define the cup product:}

With the cochains maps $g^i, \tilde g^i$,
  we are ready to compose them together with the cup product on $\cC(\tilde X, \tilde F)$
  to form the cup product on $\cC(X, F)$,
\begin{equation}
  \cC^i(X, \cF_1) \times \cC^j(X, \cF_2)
  \xrightarrow{\tilde g^{1,i} \times \tilde g^{2,j}} \cC^i(\tilde X, \tilde \cF_1) \times \cC^j(\tilde X, \tilde \cF_2)
  \xrightarrow{\cup} \cC^{i+j}(\tilde X, \tilde \cF_1 \odot \tilde \cF_2)
  \xrightarrow{g^{i+j}} \cC^{i+j}(X, \cF_1 \odot \cF_2).
\end{equation}
We use the property that $\tilde{\cF_1} \odot \tilde{\cF_2} = \widetilde{\cF_1 \odot \cF_2}$ in the last map.
Since the cochain maps commute with the coboundary operators,
  this induces the cup product for the cohomology classes
  $H^i(X, \cF_1) \times H^j(X, \cF_2)
  \xrightarrow{\cup} H^{i+j}(X, \cF_1 \odot \cF_2)$.

\subsection{A local condition that allows transversal CCZ operations for sheaf codes}

We now use the cup product to prove our main theorem.

\begin{theorem}
  Let $Q(X, \{C_{1,\sigma}\}_{\sigma \in X(t-1)}, \ell_1)$,
    $Q(X, \{C_{2,\sigma}\}_{\sigma \in X(t-1)}, \ell_2)$,
    $Q(X, \{C_{3,\sigma}\}_{\sigma \in X(t-1)}, \ell_3)$
    be three quantum codes $Q_1, Q_2, Q_3$
    such that $\ell_1 + \ell_2 + \ell_3 = t$.
  Let $\cF_i$ be the sheaf $\cF(X, \{C_{i,\sigma}\}_{\sigma \in X(t-1)})$.
  We assume that the local sections satisfy $(\cF_i)_\tau = \FF_q$ for $\tau \in X(t)$.
  We define a trilinear map
    $f: \cC^{\ell_1}(X, \cF_1) \times \cC^{\ell_2}(X, \cF_2) \times \cC^{\ell_3}(X, \cF_3) \to \FF_q$ by
  \begin{equation}
    f(\a_1, \a_2, \a_3) = \sum_{\tau \in X(t)} ((\a_1 \cup \a_2) \cup \a_3)(\tau).
  \end{equation}

  If $C_{1,\sigma} \odot C_{2,\sigma} \odot C_{3,\sigma}$ is contained within the parity check code for every $\sigma \in X(t-1)$,
    then $(Q_1, Q_2, Q_3, f)$ forms a CCZ code.
\end{theorem}

Note that $(\a_1 \cup \a_2) \cup \a_3$
  lies in $\cC^{\ell_1+\ell_2+\ell_3}(X, \cF_1 \odot \cF_2 \odot \cF_3)$,
  which can be viewed as a function in $X(t) \to \FF_q$.

\begin{proof}
  All we need to show is that $f$ is defined on the cohomology classes.
  Given $\z_1, \z_2, \z_3$ and $\b_1, \b_2, \b_3$ with
    $\z_i \in Z^{\ell_i}(X, \cF_i)$ and $\b_i \in B^{\ell_i}(X, \cF_i)$.
  Because the cup product is defined on the cohomology classes,
    we have
    $(\z'_1 \cup \z'_2) \cup \z'_3 - (\z_1 \cup \z_2) \cup \z_3 = \b$
    for some $\b \in B^{t}(X, \cF_1 \odot \cF_2 \odot \cF_3)$,
    where $\z'_i = \z_i + \b_i$.
  This implies
    \begin{equation}
      f(\z'_1, \z'_2, \z'_3) - f(\z_1, \z_2, \z_3) = \sum_{\tau \in X(t)} \b(\tau).
    \end{equation}
  Let $\beta = \delta \alpha$ with $\alpha \in \cC^{t-1}(X, \cF_1 \odot \cF_2 \odot \cF_3)$.
  We have
  \begin{equation}
    \sum_{\tau \in X(t)} \beta(\tau)
    = \sum_{\tau \in X(t)} \delta \alpha(\tau)
    = \sum_{\tau \in X(t)} \sum_{\sigma \precdot \tau} \alpha(\sigma)|_\tau
    = \sum_{\sigma \in X(t-1)} \sum_{\tau \succdot \sigma} \alpha(\sigma)|_\tau
    = 0.
  \end{equation}
  In the last equality,
    we use $\{\alpha(\sigma)|_\tau\}_{\tau \succdot \sigma} \in C_{1,\sigma} \odot C_{2,\sigma} \odot C_{3,\sigma}$
    and $\sum_{\tau \succdot \sigma} \alpha(\sigma)|_\tau = 0$
    which follows from the assumption that $C_{1,\sigma} \odot C_{2,\sigma} \odot C_{3,\sigma}$ is contained within the parity check code.
\end{proof}

More generally, we can obtain multiple CCZ operations,
  by taking four codes,
  $Q(X, \{C_{i,\sigma}\}_{\sigma \in X(t-1)}, \ell_i)$ for $i = 1, 2, 3, 4$,
  where $\ell_1 + \ell_2 + \ell_3 + \ell_4 = t$.
  (It is ok for the last code to be classical with $\ell_4 = 0$.)
We assume that for every $\sigma \in X(t-1)$,
  $C_{1,\sigma} \odot C_{2,\sigma} \odot C_{3,\sigma} \odot C_{4,\sigma}$ is contained within the parity check code.
We can then define trilinear maps $f_{\z_4}$ for every
  $\z_4 \in Z^{\ell_4}(X, \{C_{4,\sigma}\}_{\sigma \in X(t-1)})$ as
\begin{equation}
  f_{\z_4}(\a_1,\a_2,\a_3) = \sum_{\tau \in X(t)} (((\a_1 \cup \a_2) \cup \a_3) \cup \z_4)(\tau).
\end{equation}
Similarly, one can show that $f_{\z_4}$ is defined on the cohomology classes.
This is related to a known fact that
  a non-targeted CCCZ operation can induce targeted CCZ operations.

\subsection{Abstract discussions}\label{sec:abstract}

We make some abstract remarks
  and extend the cup product to the case where $\cF_\tau$ for $\tau \in X(t)$ is not necessarily $\FF_q$.

In mathematics,
  it is understood that
  given sheaves $\cF_1, \cF_2$ of vector spaces over $\FF_q$ on a topological space $X$,
  there is the cup product
\begin{equation}\label{eq:general-cup}
  H^i(X, \cF_1) \times H^j(X, \cF_2) \xrightarrow{\cup} H^{i+j}(X, \cF_1 \otimes \cF_2).
\end{equation}
In the context
  where the the value at each point is in $\FF_q$,
  this corresponds to the entrywise product we considered above,
  $(f, g) \mapsto fg$ where $fg(x) = f(x) g(x)$.
Thus, our result on cellular complexes
  can be seen as an analog of this well-known result from traditional sheaf theory.

This indicates that our result can be extended.
In particular, given two sheaves $\cF_1, \cF_2$ on a cellular complex $X$,
  we can construct the sheaf $\cF_1 \otimes \cF_2$
  defined by $(\cF_1 \otimes \cF_2)_\sigma(\tau) = \cF_{1,\sigma}(\tau) \otimes \cF_{2,\sigma}(\tau)$
  for every $\sigma \in X(t-1)$ and $\tau \in X_{\ge \sigma}(t)$.
By following a similar proof strategy,
  we can derive the cup product
  $H^i(X, \cF_1) \times H^j(X, \cF_2) \xrightarrow{\cup} H^{i+j}(X, \cF_1 \otimes \cF_2)$.

\section*{Acknowledgements}
TCL was supported in part by funds provided by the U.S. Department of Energy (D.O.E.) under the cooperative research agreement DE-SC0009919 and by the Simons Collaboration on Ultra-Quantum Matter, which is a grant from the Simons Foundation (652264 JM).

\bibliographystyle{unsrt}
\bibliography{references.bib}

\appendix

\section{Locally acyclic sheaves}\label{sec:locally-acyclic-sheaves}

For the purpose of studying sheaf cohomology,
  it is natural to demand that the cohomology vanishes on each cell.
This is essential to obtain Poincaré duality
  which will be discussed in the next section.
We denote $X_{\le \sigma}$ as the cellular complex formed by $\sigma$,
  and $\cF_{\le \sigma}$ as the restriction of the sheaf where
  $\cF_{\le \sigma}(\rho) = \cF_\rho$ for all $\rho \preceq \sigma$.
\begin{definition}
  Given a $t$ dimensional cellular complex $X$ and and a sheaf $\cF$ on $X$.
  We say $\cF$ is locally acyclic if for each $i$-cell $\sigma$,
  $H^j(X_{\le \sigma}, \cF_{\le \sigma}) = 0$ for $0 < j \le t-i$.
\end{definition}
For simplicity, we will denote $H^j(X_{\le \sigma}, \cF_{\le \sigma})$ as $H^j(\sigma, \cF)$.
More explicitly, $\cF$ is locally acyclic
  iff for each $i$-cell $\sigma$, the chain complex
\begin{equation}
  \prod_{\sigma_{i+1} \in X_{\ge \sigma}(i+1)} \cF_{\sigma_{i+1}}
  \to \prod_{\sigma_{i+2} \in X_{\ge \sigma}(i+2)} \cF_{\sigma_{i+2}}
  \to ...
  \to \prod_{\sigma_{t} \in X_{\ge \sigma}(t)} \cF_{\sigma_{t}}
\end{equation}
is exact.

The acyclic condition is often essential for estabilishing a meaningful cohomology theory,
  as it determine which structures are considered trivial.
For example, in the case of topological spaces,
  if $X$ is contractible, then $H^i(X, \FF_q) = 0$ for all $i > 0$.
Similarly, for schemes,
  if $X$ is an affine scheme and $\cF$ is a quasicoherent sheaf,
  then $H^i(X, \cF) = 0$ for all $i > 0$.

In the context of quantum codes,
  the acyclic condition naturally arises because it implies that
  every cocycle $\z$ supported on $X_{\le \sigma}$
  must be a coboundary $\z = \delta \a$ for some $\a$ supported on $X_{\le \sigma}$.
If this condition is not satisfied,
  it is almost saying that $\z$ is a nontrivial logical operator
  with a constant weight.
This is undesireable, since our goal is to construct LDPC codes with large code distances.

Another way to view the locally acyclic condition is that
  this is analogous to the condition of a Leray cover over topological spaces.
The Leray cover implies that the associate Čech cohomology
  is isomorphic to the sheaf cohomology.
This is not essential and we will not dwell on this point further.

Currently, all sheaves used for code constructions,
  such as those in \cite{panteleev2021asymptotically, dinur2021locally, leverrier2022quantum, dinur2022good, dinur2023new, dinur2024expansion},
  are locally acyclic.
Most of these constructions achieve local acyclicity by employing tensor codes.
The only exception is \cite{dinur2023new}, which uses Reed-Solomon and Reed-Muller codes.

\section{Poincaré duality for cellular complexes}\label{sec:poincare}

\begin{theorem}
  Given a $t$ dimensional cellular complex $X$ and and a sheaf $\cF$ on $X$.
  If $\cF$ is locally acyclic, for $0 \le i \le t-1$, we have
  \begin{equation}
    H_{t-i}(X, \cF) \cong H^i(X, \cF^\perp).
  \end{equation}
\end{theorem}

\begin{proof}
  Our proof utilizes a spectral sequence,
    similar to the proof strategy used in Thm 8.9 of Section 14 in \cite{bott2013differential}.
  This strategy often arises when working with Čech cohomology
    which is closely related to our setting.

  The idea is to construct a double complex
    $(K^{\bullet,\bullet}, \delta, d)$
    where $\delta: K^{p,q} \to K^{p+1,q}, d: K^{p,q} \to K^{p,q+1}$
    and $d^2 = 0, \delta^2 = 0, d \delta = \delta d$.
  Let $D = d + \delta$.
  Notice that $D^2 = 0$\footnote{
    We again avoid the discussion of the negative signs by working in characteristic $2$.
    },
    which defines a chain complex $(K^{\bullet}, D)$,
    where $K^k = \bigoplus_{p+q = k} K^{p,q}$.

  It is known that one can compute a sequence of complexes $\{(E_r, d_r)\}_{r=0}^\infty$,
    where each page $E_r$ is the cohomology of the previous page $E_{r-1}$.
  The key result of spectral sequences is that, if the sequence converges,
    the resulting complex corresponds to the cohomology group of $(K^{\bullet}, D)$.
  For a more detailed statement, we refer to Section 14 in \cite{bott2013differential}.

  In our use of the spectral sequence,
    we alter the roles of $\delta$ and $d$,
    which provides an alternative way to compute the cohomology group of $(K^{\bullet}, D)$.
  The desired result follows from the fact that both computations lead to the same cohomology group.

  We now describe the double complex.
  For $0 \le p \le r \le t$, let
  \begin{equation}
    K^{p,q} = \prod_{\sigma \in X(p)} \cC^r(\sigma, \cF)
  \end{equation}
  where $r = t-q$ and $\cC^r(\sigma, \cF) = \prod_{\tau \in X_{\ge \sigma}(r)} \cF_\tau$.
  Otherwise, $K^{p,q} = 0$.

  For $0 \le p < r \le t$, let $\delta: K^{p,q} \to K^{p+1,q}$ be the map
  \begin{equation}
    (\delta \a)(\tau) = \sum_{\sigma \precdot \tau} \a(\sigma)
  \end{equation}
  where $\tau \in X(p)$.
  Otherwise, $\delta$ is the zero map.

  For $0 \le p < r \le t$, let $d: K^{p,q} \to K^{p,q+1}$ be the map
  \begin{equation}
    (d \a)(\tau) = \partial (\a(\tau))
  \end{equation}
  where $\partial: \cC^r(\tau, \cF) \to \cC^{r-1}(\tau, \cF)$
    is the boundary map
    where $(\partial \b) (\pi) = \sum_{\rho \succdot \pi} \b (\rho)$.
  Otherwise, $d$ is the zero map.

  It is straightforward that $d^2 = 0, \delta^2 = 0$.
  For $0 \le p < r \le t$ with $r - p \ge 2$
    and $\alpha \in K^{p,q}$
    we have
    \begin{equation}
      (d \delta \alpha)(\tau)(\pi)
      = \partial (d \alpha (\tau)) (\pi)
      = \sum_{\rho \succdot \pi} d \alpha (\tau)(\rho)
      = \sum_{\rho \succdot \pi} \sum_{\sigma \precdot \tau} \alpha (\sigma)(\rho)
    \end{equation}
    and
    \begin{equation}
      (\delta d \alpha)(\tau)(\pi)
      = \sum_{\sigma \precdot \tau} (d \alpha)(\sigma)(\pi)
      = \sum_{\sigma \precdot \tau} \partial(\alpha(\sigma))(\pi)
      = \sum_{\sigma \precdot \tau} \sum_{\rho \succdot \pi} \alpha(\sigma)(\rho)
    \end{equation}
    which implies $d \delta \alpha = \delta d \alpha$.
  Otherwise $d \delta = \delta d$ as they are both zero maps.
  We have now established that $(K^{\bullet,\bullet}, \delta, d)$ forms a double complex.

  We are ready to discuss the spectral sequence.
  We will state the resulting pages in the spectral sequence,
    and then justify each step.
  The double complex forms the zero-th page
  \begin{equation}
    \centering
    E_0 = \begin{tikzpicture}[xscale={3},baseline={(X.base)}]
      \draw (0,4)node{$\prod_{\sigma \in X(0)} \cC^0(\sigma, \cF)$} (1,4)node{0} (2,4)node{0} (3,4)node{$\cdots$} (4,4)node{0};
      \draw (0,3)node{$\prod_{\sigma \in X(0)} \cC^1(\sigma, \cF)$} (1,3)node{$\prod_{\sigma \in X(1)} \cC^1(\sigma, \cF)$} (2,3)node{0} (3,3)node{$\cdots$} (4,3)node{0};
      \draw (0,2)node(X){$\prod_{\sigma \in X(0)} \cC^2(\sigma, \cF)$} (1,2)node{$\prod_{\sigma \in X(1)} \cC^2(\sigma, \cF)$} (2,2)node{$\prod_{\sigma \in X(2)} \cC^2(\sigma, \cF)$} (3,2)node{$\cdots$} (4,2)node{0};
      \draw (0,1)node{$\vdots$} (1,1)node{$\vdots$} (2,1)node{$\vdots$} (3,1)node{$\vdots$} (4,1)node{$\vdots$};
      \draw (0,0)node{$\prod_{\sigma \in X(0)} \cC^t(\sigma, \cF)$} (1,0)node{$\prod_{\sigma \in X(1)} \cC^t(\sigma, \cF)$} (2,0)node{$\prod_{\sigma \in X(2)} \cC^t(\sigma, \cF)$} (3,0)node{$\cdots$} (4,0)node{$\prod_{\sigma \in X(t)} \cC^t(\sigma, \cF)$};
      \draw[->] (-0.6,-0.5)--(4.5,-0.5)node[anchor=north]{$p$};
      \draw[->] (-0.6,-0.5)--(-0.6,4.5)node[anchor=east]{$q$};
    \end{tikzpicture}
  \end{equation}

  The first spectral sequence has
  \begin{equation}
    \centering
    E_1 = H_{\delta} E_0 = \begin{tikzpicture}[xscale={2},baseline={(X.base)}]
      \draw (0,4)node{$\cC_0(X, \cF)$} (1,4)node{0} (2,4)node{0} (3,4)node{$\cdots$} (4,4)node{0};
      \draw (0,3)node{$\cC_1(X, \cF)$} (1,3)node{$?$} (2,3)node{0} (3,3)node{$\cdots$} (4,3)node{0};
      \draw (0,2)node(X){$\cC_2(X, \cF)$} (1,2)node{0} (2,2)node{$?$} (3,2)node{$\cdots$} (4,2)node{0};
      \draw (0,1)node{$\vdots$} (1,1)node{$\vdots$} (2,1)node{$\vdots$} (3,1)node{$\vdots$} (4,1)node{$\vdots$};
      \draw (0,0)node{$\cC_t(X, \cF)$} (1,0)node{0} (2,0)node{0} (3,0)node{$\cdots$} (4,0)node{$?$};
      \draw[->] (-0.6,-0.5)--(4.5,-0.5)node[anchor=north]{$p$};
      \draw[->] (-0.6,-0.5)--(-0.6,4.5)node[anchor=east]{$q$};
    \end{tikzpicture}
  \end{equation}
  and
  \begin{equation}
    \centering
    E_\infty = E_2 = H_d H_{\delta} E_0 = \begin{tikzpicture}[xscale={2},baseline={(X.base)}]
      \draw (0,4)node{$H_0(X, \cF)$} (1,4)node{0} (2,4)node{0} (3,4)node{$\cdots$} (4,4)node{0};
      \draw (0,3)node{$H_1(X, \cF)$} (1,3)node{$?$} (2,3)node{0} (3,3)node{$\cdots$} (4,3)node{0};
      \draw (0,2)node(X){$H_2(X, \cF)$} (1,2)node{0} (2,2)node{$?$} (3,2)node{$\cdots$} (4,2)node{0};
      \draw (0,1)node{$\vdots$} (1,1)node{$\vdots$} (2,1)node{$\vdots$} (3,1)node{$\vdots$} (4,1)node{$\vdots$};
      \draw (0,0)node{$H_t(X, \cF)$} (1,0)node{0} (2,0)node{0} (3,0)node{$\cdots$} (4,0)node{$?$};
      \draw[->] (-0.6,-0.5)--(4.5,-0.5)node[anchor=north]{$p$};
      \draw[->] (-0.6,-0.5)--(-0.6,4.5)node[anchor=east]{$q$};
    \end{tikzpicture}
  \end{equation}
  Even though some of the values on the diagonals is unknown,
    we have $H^i(K, D) = H_{t-i}(X, \cF)$ for $0 \le i \le t-1$.

  The second spectral sequence has
  \begin{equation}
    \centering
    E_1 = H_d E_0 = \begin{tikzpicture}[xscale={2},baseline={(X.base)}]
      \draw (0,4)node{0} (1,4)node{0} (2,4)node{0} (3,4)node{$\cdots$} (4,4)node{0};
      \draw (0,3)node{0} (1,3)node{0} (2,3)node{0} (3,3)node{$\cdots$} (4,3)node{0};
      \draw (0,2)node(X){0} (1,2)node{0} (2,2)node{0} (3,2)node{$\cdots$} (4,2)node{0};
      \draw (0,1)node{$\vdots$} (1,1)node{$\vdots$} (2,1)node{$\vdots$} (3,1)node{$\vdots$} (4,1)node{$\vdots$};
      \draw (0,0)node{$\cC^0(X, \cF^\perp)$} (1,0)node{$\cC^1(X, \cF^\perp)$} (2,0)node{$\cC^2(X, \cF^\perp)$} (3,0)node{$\cdots$} (4,0)node{$\cC^t(X, \cF^\perp)$};
      \draw[->] (-0.6,-0.5)--(4.5,-0.5)node[anchor=north]{$p$};
      \draw[->] (-0.6,-0.5)--(-0.6,4.5)node[anchor=east]{$q$};
    \end{tikzpicture}
  \end{equation}
  and
  \begin{equation}
    \centering
    E_\infty = E_2 = H_\delta H_d E_0 = \begin{tikzpicture}[xscale={2},baseline={(X.base)}]
      \draw (0,4)node{0} (1,4)node{0} (2,4)node{0} (3,4)node{$\cdots$} (4,4)node{0};
      \draw (0,3)node{0} (1,3)node{0} (2,3)node{0} (3,3)node{$\cdots$} (4,3)node{0};
      \draw (0,2)node(X){0} (1,2)node{0} (2,2)node{0} (3,2)node{$\cdots$} (4,2)node{0};
      \draw (0,1)node{$\vdots$} (1,1)node{$\vdots$} (2,1)node{$\vdots$} (3,1)node{$\vdots$} (4,1)node{$\vdots$};
      \draw (0,0)node{$H^0(X, \cF^\perp)$} (1,0)node{$H^1(X, \cF^\perp)$} (2,0)node{$H^2(X, \cF^\perp)$} (3,0)node{$\cdots$} (4,0)node{$H^t(X, \cF^\perp)$};
      \draw[->] (-0.6,-0.5)--(4.5,-0.5)node[anchor=north]{$p$};
      \draw[->] (-0.6,-0.5)--(-0.6,4.5)node[anchor=east]{$q$};
    \end{tikzpicture}
  \end{equation}
  We have $H^i(K, D) = H^i(X, \cF^\perp)$ for $0 \le i \le t$.

  By comparing the results from the two spectral sequences,
    we have the desired result
    $H_{t-i}(X, \cF) \cong H^i(X, \cF^\perp)$
    for $0 \le i \le t-1$.

  It remains to show the following statements:
  \begin{enumerate}
    \item $0 \to \cC_0(X, \cF) \xrightarrow{f} \prod_{\sigma \in X(0)} \cC^0(\sigma, \cF) \to 0$ is exact.
    \item $0 \to \cC_i(X, \cF) \xrightarrow{g} \prod_{\sigma \in X(0)} \cC^i(\sigma, \cF) \xrightarrow{\delta} \prod_{\sigma \in X(1)} \cC^i(\sigma, \cF)$ is exact for $1 \le i \le t$.
    \item $\prod_{\sigma \in X(i)} \cC^{i+1}(\sigma, \cF) \xrightarrow{d} \prod_{\sigma \in X(i)} \cC^i(\sigma, \cF) \to 0$ is exact for $0 \le i \le t-1$.
    \item $0 \to \cC^t(X, \cF^\perp) \xrightarrow{} \prod_{\sigma \in X(t)} \cC^t(\sigma, \cF) \to 0$ is exact.
    \item $0 \to \cC^i(X, \cF^\perp) \xrightarrow{h} \prod_{\sigma \in X(i)} \cC^t(\sigma, \cF) \xrightarrow{d} \prod_{\sigma \in X(i)} \cC^{t-1}(\sigma, \cF)$ is exact for $0 \le i \le t-1$.
  \end{enumerate}
  We will specify the maps $f, g, h$.

  We show the first statement.
  Notice that $\cC^0(\sigma, \cF) = \cF_\sigma$.
  Let $f$ be the map $\a \mapsto \prod_{\sigma \in X(0)} \a(\sigma)$.
  It is straightforward to check that $f$ is an isomorphism.
  (Recall that in our context $\cC_0(X, \cF) = \cC^0(X, \cF)$.)

  We show the second statement.
  Let $g$ be the map $\a \mapsto \prod_{\sigma \in X(0)} \{\a(\tau)\}_{\tau \in X_{\ge \sigma}(i)}$.
  It is clear that $g$ is injective.
  It remains to show $\im g = \ker \delta$,
    specifically that $\im g \supseteq \ker \delta$.
  Recall that for $\b \in \prod_{\sigma \in X(0)} \cC^i(\sigma, \cF)$,
    $(\delta \b)(\rho) = \sum_{\sigma \precdot \rho} \b(\sigma)$.
  In this case, $\rho$ is a 1-cell,
    so there are exactly two 0-cells $\sigma, \sigma'$ with $\sigma, \sigma' \precdot \rho$.
  This implies that the functions $\b(\sigma)$ defined on $\tau \in X_{\ge \sigma}(i)$
    and $\b(\sigma')$ defined on $\tau \in X_{\ge \sigma'}(i)$
    agree on their overlap $\tau \in X_{\ge \rho}(i)$.
  Since the functions $\{\b(\sigma)\}_{\sigma \in X(0)}$ are consistent,
    we can glue them together to form a global section $\a$ defined on all of $X(i)$.
  It is clear that $\b = f(\a)$,
    and thus $\im g \supseteq \ker \delta$.

  We show the third statement.
  It suffices to show that
    $\cC^i(\sigma, \cF) \to \cC^{i+1}(\sigma, \cF)$,
    i.e. $\cF_\sigma \to \prod_{\sigma_{i+1} \in X_{\ge \sigma}(i+1)} \cF_{\sigma_{i+1}}$,
    is injective.
  This follows directly from the injectivity axiom on $\sigma$.

  We show the fourth statement.
  The statement is true as both $\cC^t(X, \cF^\perp)$ and $\prod_{\sigma \in X(t)} \cC^t(\sigma, \cF)$
  are the space of functions defined on $X(t)$,
    $\cC^t(X, \cF^\perp) = \prod_{\sigma \in X(t)} \cF^\perp_\sigma \cong \prod_{\sigma \in X(t)} \cF_\sigma = \prod_{\sigma \in X(t)} \cC^t(\sigma, \cF)$.

  We show the fifth statement.
  Notice that $\cC^i(X, \cF^\perp) = \prod_{\sigma \in X(i)} \cF^\perp_\sigma$.
  So all we need is to show is the chain complex
    $0 \to \cF^\perp_\sigma \xrightarrow{e_\sigma} C^t(\sigma, \cF) \xrightarrow{\partial} C^{t-1}(\sigma, \cF)$ is exact for each $\sigma \in X(i)$,
    where $e = \prod_{\sigma \in X(i)} e_\sigma$
      and $e_\sigma$ maps $c$ to $\{\res_{\sigma, \tau} c\}_{\tau \in X_{\ge \sigma}(t)}$.
  In fact, the exactness of such chain complex
    was discussed in \Cref{claim:poincare-0},
    where we take $\sigma$ as the cellular complex $X$.
  This concludes the proof for the case where $\cF$ is locally acyclic.

  For statements 3, 4, 5, it is helpful to keep in mind the following long exact sequence for each $i$-cell $\sigma$:
  \begin{equation}
    0
    \to \cF_\sigma
    \to \prod_{\sigma_{i+1} \in X_{\ge \sigma}(i+1)} \cF_{\sigma_{i+1}}
    \to \prod_{\sigma_{i+2} \in X_{\ge \sigma}(i+2)} \cF_{\sigma_{i+2}}
    \to ...
    \to \prod_{\sigma_{t} \in X_{\ge \sigma}(t)} \cF_{\sigma_{t}}
    \to \cF_\sigma^\perp
    \to 0
  \end{equation}
  which can also be written as
  \begin{equation}\label{eq:long-exact-sequence}
    0
    \to \cC^i(\sigma, \cF)
    \to \cC^{i+1}(\sigma, \cF)
    \to \cC^{i+2}(\sigma, \cF)
    \to ...
    \to \cC^{t}(\sigma, \cF)
    \to \cF_\sigma^\perp
    \to 0.
  \end{equation}
\end{proof}

\end{document}

%% file: macros.tex
\newtheorem{theorem}{Theorem}[section]

\newtheorem{lemma}[theorem]{Lemma}
\newtheorem{definition}[theorem]{Definition}

\newtheorem{example}[theorem]{Example}
\newtheorem{remark}[theorem]{Remark}

\newtheorem{claim}[theorem]{Claim}

\newcommand{\FF}{\mathbb{F}}

\newcommand{\ZZ}{\mathbb{Z}}

\newcommand{\cC}{\mathcal{C}}
\newcommand{\cD}{\mathcal{D}}
\newcommand{\cF}{\mathcal{F}}
\newcommand{\cG}{\mathcal{G}}
\newcommand{\cN}{\mathcal{N}}

\newcommand{\ol}{\overline}

\newcommand{\interior}[1]{%
  {\kern0pt#1}^{\mathrm{o}}%
}

\renewcommand{\>}{\rangle}

\let\ker\relax
\let\im\relax
\DeclareMathOperator{\ker}{ker}
\DeclareMathOperator{\im}{im}

\DeclareMathOperator{\Span}{span}
\DeclareMathOperator{\CCZ}{CCZ}

\DeclareMathOperator{\Aut}{Aut}

\DeclareMathOperator{\poly}{poly}
\DeclareMathOperator{\tr}{tr}

\DeclareMathOperator{\type}{type}
\DeclareMathOperator{\res}{res}
\DeclareMathOperator{\cores}{co-res}
\DeclareMathOperator{\Enc}{Enc}

\def\F2{\mathbb{F}_2}

\renewcommand{\a}{\alpha}
\renewcommand{\b}{\beta}
\renewcommand{\c}{\gamma}
\newcommand{\z}{\zeta}
\newcommand{\id}{\textnormal{id}}